\documentclass[a4paper,11pt]{article}
\usepackage{fullpage}
\usepackage{amsthm, amsfonts,mathrsfs,amssymb,amsmath,clrscode,multirow,graphicx,amsmath,epic,eepic}
\usepackage{graphicx}

\newcommand{\braket}[1]{\langle #1 \rangle}

\newcommand{\Int}{\mathbb Z}
\newcommand{\degen}{\unlhd}
\newcommand{\Nat}{\mathbb N}

\newcommand{\Real}{\mathbb R}

\newcommand{\ceil}[1]{\left\lceil #1 \right\rceil}
\newcommand{\mybar}[1]{\lambda}
\newcommand{\floor}[1]{\left\lfloor #1 \right\rfloor}
\newcommand{\sfloor}[1]{\lfloor #1 \rfloor}

\newcommand{\Cc}{\mathscr{C}}

\newcommand{\poly}{\mathrm{poly}}

\newcommand{\var}{\mathrm{var}}
\newcommand{\supp}{\mathrm{supp}}
\newcommand{\suppc}{\mathrm{supp_c}}

\newcommand{\E}{\mathbb{E}}
\newcommand{\Nn}{\mathcal{N}}

\newcommand{\field}{\mathbb{F}}

\newtheorem{theorem}{Theorem}[section]

\newtheorem{proposition}{Proposition}[section]
\newtheorem{definition}{Definition}[section]

\newtheorem{fact}{Fact}
\newtheorem{lemma}{Lemma}[section]
\newenvironment{proof-sketch}{\trivlist\item[]\emph{Brief proof sketch}:}%
{\unskip\nobreak\hskip 1em plus 1fil\nobreak$\Box$
\parfillskip=0pt%
\endtrivlist}
\begin{document}
\author{%
Fran{\c c}ois Le Gall \\ 
     Department of Computer Science\\
     Graduate School of Information Science and Technology\\
    The University of Tokyo\\
    \texttt{legall@is.s.u-tokyo.ac.jp} }
\title{Faster Algorithms for Rectangular Matrix Multiplication}
\date{}
\maketitle
\begin{abstract} 
Let $\alpha$ be the maximal value such that the product of an $n\times n^\alpha$ matrix by an $n^\alpha\times n$ matrix can be computed with $n^{2+o(1)}$ arithmetic operations. In this paper we show that $\alpha>0.30298$, which improves the previous record $\alpha>0.29462$ by Coppersmith (Journal of Complexity, 1997). More generally, we construct a new algorithm for multiplying an $n\times n^k$ matrix by an $n^k\times n$ matrix, for any value $k\neq 1$. The complexity of this algorithm is better than all known algorithms for rectangular matrix multiplication. In the case of square matrix multiplication (i.e., for $k=1$), we recover exactly the complexity of the algorithm 
by Coppersmith and Winograd (Journal of Symbolic Computation, 1990).

These new upper bounds can be used to improve the time complexity of several known algorithms that rely on rectangular matrix multiplication. For example, we directly obtain a $O(n^{2.5302})$-time algorithm for the all-pairs shortest paths problem over directed graphs with small integer weights, improving over the $O(n^{2.575})$-time algorithm by Zwick (JACM 2002), and also improve the time complexity of sparse square matrix multiplication.
\end{abstract}
\section{Introduction}

\paragraph{Background.}
Matrix multiplication is one of the most fundamental problems in computer science and mathematics. 
Besides the fact that several computational problems in linear algebra can be reduced to the computation of the product 
of two matrices, the complexity of  matrix multiplication also arises as a bottleneck in a multitude of 
other computational tasks (e.g., graph algorithms). 
The standard method for multiplying two $n\times n$ matrices uses $O(n^3)$ arithmetic operations.
Strassen showed in 1969 that this trivial algorithm is not optimal, and gave a
algorithm that uses only $O(n^{2.808})$ arithmetic operations. 
This has been the beginning  of a long story of improvements that lead to the upper bound
$O(n^{2.376})$ by Coppersmith and Winograd \cite{Coppersmith+90}, which has been further
improved to $O(n^{2.3727})$
very recently by Vassilevska Williams~\cite{WilliamsSTOC12}. 
Note that all the above complexities refer to the number of 
arithmetic operations involved, but naturally the same upper bounds hold for the time complexity
as well when each arithmetic operation can be done in negligible time (e.g., in $\poly(\log n)$ time).

Finding the optimal value of the exponent of square matrix multiplication 
is naturally
one of the most important open problems in algebraic complexity.
It is widely believed that 
the product of two $n\times n$ matrices can be computed with $O(n^{2+\epsilon})$ arithmetic 
operations
for any constant $\epsilon>0$. Several conjectures, including conjectures about combinatorial
structures~\cite{Coppersmith+90} and about group theory \cite{Cohn+FOCS03,Cohn+FOCS05}, 
would, if true, lead to this result 
(see also~\cite{Alon+CCC12} for recent work on these conjectures).
Another way to interpret this open problem is by considering the multiplication of an 
$n\times m$ matrix by an $m\times n$ matrix. Suppose that the matrices are defined 
over a field. 
For any $k>0$, define the exponent of such a rectangular matrix multiplication as follows:
$$
\omega(1,1,k)=\inf\{\tau\in \Real\:|\: C(n,n,\sfloor{n^k})=O(n^\tau)\},
$$
where $C(n,n,\sfloor{n^k})$ denotes the minimum number of arithmetic operations needed to multiply an 
$n\times \sfloor{n^k}$ matrix by an $\sfloor{n^k}\times n$ matrix.
Note that, while the value $\omega(1,1,k)$ may depend on the field under consideration, 
it is known that it can depend only on the characteristic of the field \cite{Schonhage81}.
Define $\omega=\omega(1,1,1)$ and 
$\alpha=\sup\{k\:|\:\omega(1,1,k)=2\}$.
The value~$\omega$ represents the exponent of square matrix multiplication, and 
the value $\alpha$ essentially represents the largest value such that 
the product of an $n\times n^\alpha$ matrix by an $n^\alpha\times n$ matrix 
can be computed with $O(n^{2+\epsilon})$ arithmetic operations for any constant~$\epsilon$. Since 
$\omega=2$ if and only if $\alpha=1$, one possible 
strategy towards showing that $\omega=2$ is to give lower bounds on $\alpha$. 
Coppersmith~\cite{CoppersmithSICOMP82} showed in 1982 that $\alpha>0.172$. 
Then, based on the techniques developed in \cite{Coppersmith+90},
Coppersmith~\cite{Coppersmith97} improved this lower bound to $\alpha>0.29462$.
This is the best lower bound on $\alpha$ known so far.

Excepting Coppersmith's works on the value $\alpha$, 
there have been relatively few algorithms that focused 
specifically on rectangular matrix multiplication.
Since it is well known (see, e.g, \cite{Lotti+83}) that multiplying an $n\times n$ matrix by an $n\times m$ matrix, 
or an $m\times n$ matrix by an $n\times n$ matrix, can be done with the same 
number of arithmetic operations as multiplying an $n\times m$ matrix by an $m\times n$ matrix, the value
$\omega(1,1,k)$ represents the exponent of all these three types of rectangular matrix multiplications.
Note that, by decomposing the product into smaller
matrix products, it is easy to obtain (see, e.g, \cite{Lotti+83}) the following upper bound:
\begin{equation}\label{eq_Huang}
\omega(1,1,k)=\left\{
\begin{array}{ll}
2 &\textrm{ if }0\le k\le \alpha\\
2+(\omega-2)\frac{k-\alpha}{1-\alpha} &\textrm{ if }\alpha\le k\le 1.
\end{array}
\right.
\end{equation}
Lotti and Romani \cite{Lotti+83} obtained nontrivial upper bounds on $\omega(1,1,k)$
based on the seminal result by Coppersmith \cite{CoppersmithSICOMP82} and on early works on 
square matrix multiplication. 
Huang and Pan \cite{Huang+98} showed how to apply ideas 
from \cite{Coppersmith+90} to the rectangular setting and obtained the upper 
bound $\omega(1,1,2)<3.333954$, but
this approach did not lead to any upper bound better than (\ref{eq_Huang}) for $k \le 1$.
Ke, Zeng, Han and Pan~\cite{Ke+08} further improved Huang and Pan's result to 
$\omega(1,1,2)<3.2699$, by using again the approach from \cite{Coppersmith+90},
and also reported the upper bounds $\omega(1,1,0.8)<2.2356$ and $\omega(1,1,0.5356)<2.0712$, which are
better than those obtained by (\ref{eq_Huang}).
Their approach, nevertheless, did not give any improvement for the value of $\alpha$.

Besides the fact that a better understanding of $\omega(1,1,k)$ gives insights 
into the nature of  matrix multiplication and ultimately may help showing that $\omega=2$,
fast algorithms for multiplying 
an $n\times n^k$ matrix by an $n^k\times n$ with $k \neq 1$ have also a multitude of applications. 
Typical examples not directly related to linear algebra
include the construction of fast algorithms for the all-pairs shortest paths problem \cite{Alon+ESA07, Roditty+11,YusterSODA09,ZwickSTOC99,ZwickJACM02},
the dynamic computation 
of the transitive closure \cite{Demetrescu+FOCS00,Sankowski+10},
finding ancestors \cite{Czumaj+TCS07}, detecting directed cycles \cite{Yuster+SODA04}, or
computing the diameter of a graph~\cite{Yuster10}.
Rectangular matrix multiplication has also been used in computational complexity \cite{Patrascu+SODA10,WilliamsCCC11},
and to speed-up 
sparse square matrix multiplication \cite{Amossen+09,Kaplan+06,Yuster+05} or
tasks in computational geometry~\cite{Kaplan+SODA07,Kaplan+06}.
Obtaining new upper bounds on $\omega(1,1,k)$ would thus reduce the asymptotic 
time complexity of algorithms in a wide range of areas.
We nevertheless stress that such improvements are only of theoretical interest, 
since the huge constants involved in the complexity of fast matrix multiplication 
usually make these algorithms impractical.

\paragraph{Short description of the approach by Coppersmith and Winograd.}
The results \cite{Coppersmith97,Huang+98,Ke+08,Stothers10,WilliamsSTOC12} 
mentioned above are all obtained by extending the approach by Coppersmith and Winograd \cite{Coppersmith+90}.
This approach is an illustration of a
general methodology initiated in the 1970's based on the theory of bilinear and trilinear forms, through which most of the improvements 
for matrix multiplication have been obtained. 
Informally, the idea is to start with a \emph{basic construction} (some small trilinear form), and then
exploit general properties of matrix multiplication (in particular Sch\"onhage's asymptotic sum inequality~\cite{Schonhage81})
to derive an upper bound on the exponent~$\omega$ from this construction. 
The main contributions of~\cite{Coppersmith+90} consist of two parts: the discovery of new
basic constructions and the introduction of strong techniques to analyze them.  
In their paper, Coppersmith and Winograd actually present three algorithms, based on three different basic constructions.
The first basic construction (Section 6 in \cite{Coppersmith+90}) is the simplest of the three and leads to the upper bound $\omega<2.40364$.
The second basic construction (Section~7 in \cite{Coppersmith+90}), that we will refer in this paper
as $F_q$ (here $q\in \Nat$ is a parameter), leads to the upper bound $\omega<2.38719$. The third basic construction (Section 8 in \cite{Coppersmith+90}) 
is $F_q\otimes F_q$, the tensor product of two instances of $F_q$, and leads to the improved upper bound $\omega<2.375477$.

In view of the last result, it was natural to ask if taking larger tensor powers of $F_q$ as the basic construction 
leads to better bounds on $\omega$. 
The case $r=3$ was explicitly mentioned as an open problem in~\cite{Coppersmith+90} but 
did not seem to lead to any improvement.
Stothers \cite{Stothers10} and Vassilevska Williams~\cite{WilliamsSTOC12}
succeeded in analyzing the fourth tensor product $F_q^{\otimes 4}$ and obtained
a better upper bound on $\omega$, the first improvement in more that twenty years.
Vassilevska Williams further presented a general framework that enables a systematic analysis for higher tensor products of the
basis construction, and used this framework to show that $\omega<2.3727$, for the basic construction $F^{\otimes 8}$, the best upper bound obtained so far.

The algorithms for rectangular matrix multiplication \cite{Coppersmith97,Huang+98,Ke+08} already mentioned use a similar approach. 
Huang and Pan~\cite{Huang+98} obtained their improvement on $\omega(1,1,2)$ by taking the easiest of the three construction in 
\cite{Coppersmith+90} and carefully modifying the analysis to evaluate the complexity of rectangular matrix multiplication. 
Ke, Zeng, Han and Pan \cite{Ke+08} obtained their improvements similarly, but by using the second basic construction from 
\cite{Coppersmith+90} (the construction $F_q$) instead, which lead to better upper bounds. 
These approaches, 
while very natural, do not provide any nontrivial lower bounds on $\alpha$:
the upper bounds on $\omega(1,1,k)$ obtained are strictly larger than 2 even for small values of~$k$. 
In order to obtain the lower bound $\alpha>0.29462$, Coppersmith \cite{Coppersmith97} relied on
a more complex approach: the basic construction considered is still $F_q$, but several instances 
for distinct values of $q$ are combined together in a subtle way in order to keep the complexity of the resulting
algorithm small enough (i.e., not larger than $n^{2+o(1)}$). 

\paragraph{Statement of our results and discussion.}
In this paper we construct new algorithms for rectangular matrix multiplication,
by taking the tensor power $F_q\otimes F_q$ as basic construction
and analyzing this construction in the framework of rectangular matrix multiplication. 
We use these ideas to prove that $\omega(1,1,k)=2$ for any $k\le 0.3029805$, as 
stated in the following theorem.
\begin{theorem}\label{th_alpha}
For any value $k<0.3029805... $, 
the product of an $n\times n^k$ matrix by an $n^k\times n$ matrix
can be computed with $O(n^{2+\epsilon})$ arithmetic operations for any constant $\epsilon>0$.
\end{theorem}
Theorem \ref{th_alpha} shows that $\alpha>0.30298$, which improves the previous record
$\alpha>0.29462$ by Coppersmith. 
More generally, in the present work we present an algorithm for multiplying 
an $n\times n^k$ matrix by an $n^k\times n$ matrix, for any value $k$. 
We show that the complexity of this algorithm can be expressed as a (nonlinear) 
optimization problem, and use this formulation to derive upper bounds on $\omega(1,1,k)$.
Table \ref{table_results} shows the bounds we obtain
for several values of~$k$.
The bounds obtained for $0\le k\le 1$ are represented in Figure \ref{fig_results} as well. 
\begin{table}[ht]
\begin{minipage}[b]{0.33\linewidth}\centering
\begin{tabular}{ |c | c |}
  \hline
  \multirow{2}{*}{$k$} & upper bound \\
  &on $\omega(1,1,k)$\\
  \hline                      
0.30298 &  2  \\
0.31 &   2.000063 \\
0.32 &   2.000371 \\
0.33 &   2.000939 \\
0.34 &   2.001771 \\
0.35 &   2.002870 \\
0.40 &   2.012175 \\
0.45 &   2.027102 \\
0.50 &   2.046681 \\
0.5302&2.060396\\
0.55 &   2.070063 \\
\hline  
\end{tabular}
\end{minipage}
\begin{minipage}[b]{0.33\linewidth}
\centering
\begin{tabular}{ |c | c |}
  \hline
  \multirow{2}{*}{$k$} & upper bound \\
  &on $\omega(1,1,k)$\\
  \hline
  0.60 &   2.096571 \\
  0.65 &   2.125676 \\
  0.70 &  2.156959 \\
  0.75 &  2.190087 \\
0.80 & 2.224790 \\
  0.85 &  2.260830 \\
  0.90 &  2.298048 \\
0.95 &  2.336306 \\
1.00 &  2.375477 \\
1.10 &   2.456151 \\
1.20 &   2.539392 \\
  \hline  
\end{tabular}
\end{minipage}
\begin{minipage}[b]{0.33\linewidth}
\centering
\begin{tabular}{ |c | c | }
  \hline
  \multirow{2}{*}{$k$} & upper bound \\
  &on $\omega(1,1,k)$\\
  \hline    
1.30 &   2.624703 \\
1.40 &   2.711707 \\
1.50 &   2.800116 \\
1.75 &   3.025906 \\
2.00 &   3.256689 \\
2.25 &   3.490957 \\
2.50  & 3.727808 \\
3.00  & 4.207372 \\
3.50  & 4.693151  \\
4.00  &  5.180715 \\
5.00  & 6.166736 \\
  \hline  
\end{tabular}
\end{minipage}
\caption{Our upper bounds on the exponent of the multiplication of an 
$n\times n^k$ matrix by an $n^k\times n$ matrix.
\label{table_results}}
\end{table}

\begin{figure}[h!]
\begin{center}
\includegraphics[scale=0.8]{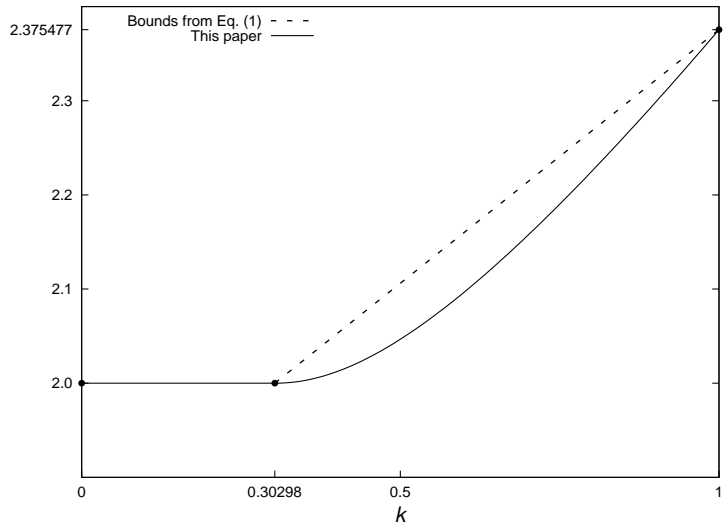}
\caption{
Our upper bounds (in plain line) on $\omega(1,1,k)$, for $0\le k\le 1$.
The dashed line represents the upper bounds on $\omega(1,1,k)$ obtained by using 
Equation (\ref{eq_Huang}) with the 
values $\alpha>0.30298$ and $\omega<2.375477$.
\label{fig_results}}
\end{center}
\end{figure}

The results of this paper can be seen as a generalization of Coppersmith-Winograd's approach to 
the rectangular setting. 
In the case of square matrix multiplication (i.e., for $k=1$), we recover naturally
the same upper bound $\omega(1,1,1)<2.375477$ as the one obtained 
in \cite{Coppersmith+90}. 
Let us mention that we can, in a rather straightforward way, 
combine our results with the upper bound $\omega<2.3727$ by 
Vassilevska Williams~\cite{WilliamsSTOC12} to obtain slightly 
improved bounds for $k\approx 1$. The idea is, very similarly to how Equation~(\ref{eq_Huang}) was obtained,
to exploit the convexity of the function $\omega(1,1,k)$.
Concretely, for any fixed value $0\le k_0<1$, the inequality
$$\omega(1,1,k)\le \omega(1,1,k_0)+(\omega-\omega(1,1,k_0))\frac{k-k_0}{1-k_0}$$
holds for any $k$ such that $k_0\le k\le 1$. This enables us to combine an upper bound on $\omega(1,1,k_0)$,
for instance one of the values in Table \ref{table_results}, with the improved upper bound $\omega<2.3727$ by
Vassilevska Williams.
Since the improvement is small and concerns only the case $k\approx 1$, 
we will not discuss it further.

For $k>0.29462$ and $k\neq 1$, the complexity of our algorithms is better than all known algorithms for rectangular matrix multiplication, including the algorithms \cite{Huang+98,Ke+08} mentioned above. 
Moreover, for $0.30298<k<1$, our new bounds are significantly better than 
what can be obtained solely from the bound $\alpha>0.30298$ and $\omega<2.375477$ through 
Equation~(\ref{eq_Huang}), as illustrated in Figure \ref{fig_results}. This suggests that non-negligible
improvements can be obtained for all applications of rectangular matrix multiplications that rely on this simple linear interpolation --- we will elaborate on this subject in Subsection \ref{sub_appl}.

Let us compare more precisely our results with those reported in \cite{Ke+08}. 
For $k=2$, we obtain $\omega(1,1,2)<3.256689$ while
Ke et al.~\cite{Ke+08} obtained $\omega(1,1,2)<3.2699$ by using the basic construction~$F_q$.
Our improvements are of the same order for the other two values ($k=0.8$ and $k=0.5356$) 
analyzed in~\cite{Ke+08}, as can be seen from Table \ref{table_results}.  
Note that the order of magnitude of the improvements here is similar to what was obtained in \cite{Coppersmith+90}
by changing
the basic construction from $F_q$ to $F_q\otimes F_q$ for square matrix multiplication, 
which led to a improvement
from $\omega<2.38719$ to $\omega<2.375477$.

A noteworthy point is that our algorithm directly leads to improved lower bounds on $\alpha$ while,
as already mentioned, to obtain a nontrivial lower bound on $\alpha$
using the basic construction~$F_q$ (as done in~\cite{Coppersmith97})
a specific methodology was needed. 
Our approach can then be considered as
a general framework to study rectangular matrix multiplication, which
leads to a unique optimization problem that gives upper bounds on $\omega(1,1,k)$ 
for any value of $k$.

\subsection{Applications}\label{sub_appl}
As mentioned in the beginning of the introduction, improvements on the time
complexity of rectangular matrix multiplication give faster algorithms for a 
multitude of computational problems. In this subsection we describe 
quantitatively the improvements that our new upper bounds imply for some of 
these problems: sparse square matrix multiplication, the all-pairs shortest paths 
problem and computing dynamically the transitive closure of a graph. 
%
\paragraph{Sparse square matrix multiplication.}
Yuster and Zwick \cite{Yuster+05} have shown how fast algorithms for 
rectangular matrix multiplication can be used to construct fast algorithms for
computing the product of two sparse square matrices (this result has been 
generalized to the product of sparse rectangular matrices 
in \cite{Kaplan+06}, and the case where the 
output matrix is also sparse has been studied in \cite{Amossen+09}). 
More precisely, let 
$M$ and $M'$ be two $n\times n$ matrices such that 
each matrix has at most $m$ non-zero entries, where $0\le m\le n^2$. 
Yuster and Zwick \cite{Yuster+05} showed that the product 
of $M$ and $M'$ can be computed in time
$$
O\left(\min(nm,n^{\omega(1,1,\lambda_m)+o(1)},n^{\omega+o(1)})\right),
$$
where $\lambda_m$ is the solution of the
equation $\lambda_m+\omega(1,1,\lambda_m)=2\log_n (m)$. 
Using the upper bounds on $\omega(1,1,k)$ of Equation 
(\ref{eq_Huang}) with the values $\alpha<0.294$ and $\omega<2.376$, 
this gives the complexity depicted in Figure~\ref{fig_sparse}. 

These upper bounds can be of course directly improved by using the new
upper bound on $\omega$ by Vassilevska Williams \cite{WilliamsSTOC12}
and the new lower bound on $\alpha$ given in the present work, 
but the improvement is small.
A more significant improvement can be obtained by using directly 
the upper bounds on $\omega(1,1,k)$ presented in Figure~\ref{fig_results}, which
gives the new upper bounds on the complexity of
sparse matrix multiplication depicted in Figure \ref{fig_sparse}.
For example, for $m=n^{4/3}$, we obtain complexity $O(n^{2.087})$,
which is better than the original upper bound $O(n^{2.1293...})$
obtained from Equation 
(\ref{eq_Huang}) with $\alpha>0.294$ and $\omega<2.376$.
Note that replacing $\omega<2.376$ with the the best known bound $\omega<2.3727$
only decreases the latter bound to $O(n^{2.1287...})$. Thus, even if the algorithms
presented in the present paper do not give any improvement on $\omega$ (i.e., for the 
product of dense square matrices), we do obtain
improvements for computing the product of two sparse square matrices.

\begin{figure}
\begin{center}
 \includegraphics[scale=0.8]{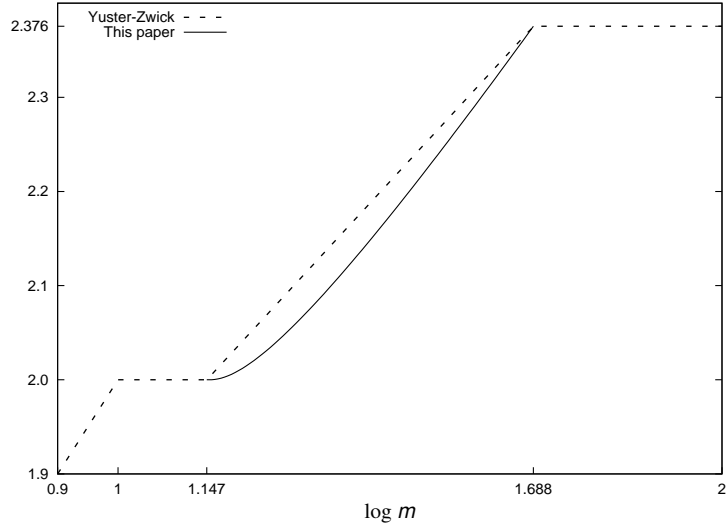}
\caption{
Upper bounds on the exponent for the multiplication two $n\times n$ matrices
with at most $m$ non-zero entries. The horizontal axis represents $\log_n(m)$. 
The dashed line represents the results by Yuster and Zwick \cite{Yuster+05} and shows
that the term $n^{\omega(1,1,\lambda_m)}$ dominates the complexity 
when $1\le \log_n(m)\le (1+\omega)/2$.
The plain line represents the improvements we obtain.
\label{fig_sparse}}
\end{center}
\end{figure}

\paragraph{Graph algorithms.}
Zwick \cite{ZwickJACM02} has shown how to use rectangular matrix multiplication to compute,
with high probability, 
the all-pairs shortest paths in weighted direct graphs where the weights are bounded integers.
The time complexity obtained is $O(n^{2+\mu+\epsilon})$, for any constant $\epsilon>0$,
where $\mu$ is the solution of the equation $\omega(1,1,\mu)=1+2\mu$. 
Using the upper bounds on $\omega(1,1,k)$ of Equation 
(\ref{eq_Huang}) with $\alpha>0.294$ and $\omega<2.376$, 
this gives $\mu<0.575$ and thus complexity $O(n^{2.575+\epsilon})$.
This reduction to rectangular matrix multiplication
is the asymptotically fastest known approach for weighted directed graphs with small integer weights.

Our results (see Table \ref{table_results}) show that $\omega(1,1,0.5302)<2.0604$, which gives the upper bound
$\mu<0.5302$. We thus obtain the following result.
\begin{theorem}
There exists a randomized algorithm that computes the shortest paths between all pairs 
of vertices in a weighted directed graph with bounded integer weights in time $O(n^{2.5302})$, 
where $n$ is the number of vertices in the graph.
\end{theorem}
Note that, even if $\omega=2$, the complexity of Zwick's algorithm is $O(n^{2.5+\epsilon})$. 
In this perspective, our improvements on the complexity of rectangular matrix multiplication offer a 
non-negligible speed-up for the all-pairs shortest paths problem in this setting.

The same approach can be used to improve several other existing graph algorithms.
Let us describe another example: algorithms for computing 
dynamically the transitive closure of a graph.
Demetrescu and Italiano \cite{Demetrescu+FOCS00} presented a randomized
algorithm for the dynamic transitive closure of directly acyclic graph with $n$ vertices 
that answers queries in $O(n^{\mu})$ time, and performs updates in $O(n^{1+\mu +\epsilon})$
time, for any $\epsilon>0$. Here $\mu$ is again the solution of the equation $\omega(1,1,\mu)=1+2\mu$. 
This was the first algorithm for this problem with subquadratic time 
complexity. This result have been generalized later to general graphs, with the same bounds, by 
Sankowski \cite{SankowskiFOCS04}. Our new upper bounds thus show the existence of an algorithm 
for the dynamic transitive closure
that answers queries in $O(n^{0.5302})$ time and performs updates in $O(n^{1.5302})$ time.

\subsection{Overview of our techniques and organization of the paper}\label{sub_overview}
Before presenting an overview of the techniques used in this paper,
we will give
an informal description of algebraic complexity theory 
(the contents of which
will be superseded by the formal presentation of these notions in Section~\ref{sec_prelim}).
In this paper we will use, for any positive integer $n$, the notation $[n]$ to represent the set $\{1,\ldots,n\}$.

\paragraph{Trilinear forms and bilinear algorithms.}
The matrix multiplication of an $m\times n$ matrix by an $n\times p$ matrix can be represented by the following
trilinear form, denoted as $\braket{m,n,p}$:
\begin{equation*}\label{eq_mm}
\braket{m,n,p}=\sum_{r=1}^m\sum_{s=1}^n\sum_{t=1}^p x_{rs}y_{st}z_{rt},
\end{equation*}
where $x_{rs}$, $y_{st}$ and $z_{rt}$ are formal variables.
This form can be interpreted as follows: the $(r,t)$-th entry of the product of an $m\times n$ matrix $M$ by an $n\times p$ 
matrix $M'$ can be obtained by setting $x_{ij}=M_{ij}$ for all $(i,j)\in[m]\times[n]$ and $y_{ij}=M'_{ij}$ for all $(i,j)\in[n]\times[p]$, 
setting $z_{rt}=1$ and setting all the other $z$-variables to zero.
One can then think of the $z$-variables as formal variables used to record the entries of the matrix product.

More generally, a trilinear form $t$ is represented as 
$$
t=\sum_{i\in A}\sum_{j\in B}\sum_{k\in C} t_{ijk}x_{i}y_{j}z_{k}.
$$
where $A,B$ and $C$ are three sets, $x_{i}$, $y_{j}$ and $z_{k}$ are formal variables and
the $t_{ijk}$'s are coefficients in a field $\field$. Note that the set of indexes for the trilinear 
form $\braket{m,n,p}$ are $A=[m]\times [n]$,  $B=[n]\times [p]$ and $C=[m]\times [p]$.

An exact (bilinear) algorithm computing $t$ corresponds to an equality of the form 
$$
t=
\sum_{\ell=1}^r \left(\sum_{i\in A}\alpha_{\ell i}x_i\right)\left(\sum_{j\in B}\beta_{\ell j}y_j\right)\left(\sum_{k\in C}\gamma_{\ell k}z_k\right)
$$
with coefficients $\alpha_{\ell i},\beta_{\ell j},\gamma_{\ell k}$ in $\field$.
The minimum number $r$ such that such a decomposition exists is called the rank of the trilinear form $t$, and denoted $R(t)$. 
The rank of a trilinear form is an upper bound on the complexity of a (bilinear) algorithm 
computing the form: 
it precisely expresses the number of multiplications needed for the computation, and it can be 
shown that the number of additions or scalar multiplications affect the cost only in a negligible way.
For any $k>0$, the quantity $\omega(1,1,k)$ can then be equivalently defined as follows:
$$
\omega(1,1,k)=\inf\{\tau\in \Real\:|\: R(\braket{n,n,\sfloor{n^k}})=O(n^\tau)\}.
$$

Approximate bilinear algorithms have been introduced to take advantage of the fact that the complexity of trilinear forms 
(and especially of matrix multiplication) may be reduced significantly by
allowing small errors in the computation.
Let $\lambda$ be an indeterminate over $\field$, and let $\field[\lambda]$ denote the 
set of all polynomials over $\field$ in $\lambda$. 
Let $s$ be any nonnegative integer.
A $\lambda$-approximate algorithm for $t$ is an equality of the 
form 
$$
\lambda^{s}t+\lambda^{s+1} \left(\sum_{i\in A}\sum_{j\in B}\sum_{k\in C} d_{ijk}x_iy_jz_k\right)=
\sum_{\ell=1}^r \left(\sum_{i\in A}\alpha_{\ell i}x_i\right)\left(\sum_{j\in B}\beta_{\ell j}y_j\right)\left(\sum_{k\in C}\gamma_{\ell k}z_k\right)
$$
for coefficients $\alpha_{\ell i},\beta_{\ell j},\gamma_{\ell k},d_{ijk}$ in $\field[\lambda]$ and some nonnegative integer $s$. 
Informally, this means that the form $t$ can be computed by determining the coefficient of $\lambda^s$ in the right hand side.
The minimum number $r$ such that such a decomposition exists is called the border rank of $t$ and 
denoted $\underline{R}(t)$. 
It is known that the border rank is an upper bound on the complexity of an algorithm that approximates the trilinear form,
and that any such approximation algorithm can be converted into an exact algorithm with essentially the same complexity.


A sum $\sum_{i}t_i$ of trilinear forms is a direct sum if the $t_i$'s do not
share variables. Informally, Sch\"onhage's asymptotic sum inequality \cite{Schonhage81} 
for rectangular matrix multiplication
states that, if the form $t$ can be converted (in the $\lambda$-approximation sense) 
into a direct sum of $c$ trilinear forms, 
each form being isomorphic
to $\braket{m,m,m^k}$, then 
$$
c\cdot m^{\omega(1,1,k)}\le \underline{R}(t).
$$
This suggests that good bounds on $\omega(1,1,k)$ can be obtained if 
the form $t$ can be used to derive many independent (i.e., not sharing any variables) 
matrix multiplications. This approach has been applied to derive almost all new bounds 
on matrix multiplication since its discovery in 1981 by Sch\"onhage.

\paragraph{Overview of our techniques.}
Our algorithm uses, as its basic construction, the trilinear form $F_q\otimes F_q$ from 
\cite{Coppersmith+90}, which can be written as a sum of fifteen terms $T_{ijk}$, for all fifteen 
nonnegative integers $i,j,k$ such that $i+j+k=4$:
$$
F_q\otimes F_q=\sum_{\begin{subarray}{c}0\le i,j,k\le 4\\ i+j+k=4\end{subarray}}T_{ijk}.
$$
If the sum were direct, then, from Sch\"onhage's asymptotic sum inequality and since an upper bound on 
the rank of $F_q\otimes F_q$ is easy to obtain,
this would reduce the problem to the analysis of each part $T_{ijk}$. 
This is unfortunately not 
the case: the $T_{ijk}$'s share variables. To solve this problem, the basic 
construction is manipulated in order to obtain a direct sum, similarly to 
\cite{Coppersmith+90}.
The first step is to take the $N$-th 
tensor product of the basis construction, where $N$ is a large integer. This gives:
$$
(F_q\otimes F_q)^{\otimes N}=\sum_{IJK} T_{IJK}
$$
where the sum is over all triples of sequences $IJK$ with $I,J,K\in \{0,1,2,3,4\}^N$ such that $I_\ell+J_\ell+K_\ell=4$ for all $\ell\in\{1,\ldots,N\}$. Here each form $T_{IJK}$ is the 
tensor product of $N$ forms $T_{ijk}$.
The sum is nevertheless not yet direct. The key idea of the next step is to zero variables in order
to remove some forms and obtain a sum where any two non-zero forms $T_{IJK}$ and $T_{I'J'K'}$
are such that $I\neq I'$, $J\neq J'$ and $K\neq K'$, which will imply that the sum is direct. 
Moreover, in order to be able to apply Sch\"onhage's asymptotic sum inequality, we would like
all the remaining $T_{IJK}$ to be isomorphic to the same rectangular matrix product
(i.e., there should exist values $m$ and $k$ such that each non-zero form~$T_{IJK}$ is 
isomorphic to the matrix multiplication $\braket{m,m,m^k}$).
From here our approach differs from 
Coppersmith and Winograd \cite{Coppersmith+90}, since our concern is upper bounds
for rectangular matrix multiplication.


We will take fifteen integers $a_{ijk}$ 
and show how the form $(F_q\otimes F_q)^{\otimes N}$
can be transformed, by zeroing variables, into a \emph{direct sum} of a large number
of forms $T_{IJK}$ in which each 
$T_{IJK}$
is such that
\begin{equation}\label{eq_T}
T_{IJK}\cong
\bigotimes_{\begin{subarray}{c}0\le i,j,k\le 4\\ i+j+k=4\end{subarray}}T_{ijk}^{\otimes a_{ijk}N}.
\end{equation}
The main difference here is that, in \cite{Coppersmith+90},
the symmetry of square multiplication implied that the $a_{ijk}$'s could be taken invariant 
under permutation of indices, with means that only four parameters ($a_{004}$, $a_{013}$, $a_{022}$ and 
$a_{112}$) needed to be considered. 
In our case, we still impose the condition $a_{ijk}=a_{ikj}$, but not more.
This reduces the number of parameters to nine:
$a_{004}$, $a_{400}$, $a_{013}$, $a_{103}$, $a_{301}$, $a_{022}$, $a_{202}$, 
$a_{112}$ and $a_{211}$.
Many nontrivial technical problems arise from this larger number of parameters. 
In particular, the equations that occur during the analysis do not have 
a unique solution and an optimization step is necessary. This is similar to the difficulties
that appeared in the analysis of the basis construction $F_q^{\otimes 4}$ (for square matrix
multiplication) done in \cite{Stothers10,WilliamsSTOC12}.
We will show that this further optimization step essentially 
imposes the additional (nonlinear) constraint 
$a_{013}a_{202}a_{112}=a_{103}a_{022}a_{211}$. 

Showing that $(F_q\otimes F_q)^{\otimes N}$ can be transformed into a direct sum of many isomorphic forms,
as claimed, was also used in previous works based on the approach by Coppersmith and Winograd.
Our setting is nevertheless more general than in \cite{Coppersmith+90}, for the following two reasons.
First, our approach is asymmetric since our parameters $a_{ijk}$ are not invariant under permutations 
of the indices. Second, the technical problems due to the presence of many parameters require more 
precise arguments. The former complication was addressed implicitly in~\cite{Coppersmith97}, and explicitly in
\cite{Huang+98,Ke+08}. The latter complication was addressed in \cite{Stothers10,WilliamsSTOC12}.
Here we need to deal with these two complications simultaneously, which involves a careful analysis. Instead
of approaching this task directly in the language of trilinear forms, we give a graph-theoretic 
interpretation of it, which will make the exposition more intuitive, and also simplify the analysis. 
Each form $T_{IJK}$ will correspond to one vertex in a graph, and an edge in the graph will represent the 
fact that two forms share one index.
We will interpret the task of converting the sum $(F_q\otimes F_q)^{\otimes N}$ 
into a direct sum (i.e., a sum where the non-zero forms $T_{IJK}$ do not share any index)
as the task of converting the graph,  by using only simple graph operations, into an edgeless subgraph. 
We will present algorithms solving this latter graph-theoretic task and show that a large edgeless 
subgraph can be obtained, which means that many forms $T_{IJK}$ not sharing any index can
be constructed from $(F_q\otimes F_q)^{\otimes N}$.

Now that we have a direct sum of many forms, each form being isomorphic to (\ref{eq_T}),
the only thing to do before applying Sch\"onhage's asymptotic sum inequality 
is to show that the form (\ref{eq_T})
is isomorphic to a direct sum of  matrix products $\braket{m,m,m^k}$.
Some of the $T_{ijk}$'s (more precisely, all the $T_{ijk}$'s except $T_{112}$, $T_{121}$ and $T_{211}$)
can be analyzed in a straightforward way, 
since they correspond to matrix products, as originally observed in \cite{Coppersmith+90}.
The forms $T_{112}$, $T_{121}$ and $T_{211}$ are delicate to analyze since they do not
correspond to matrix products. In \cite{Coppersmith+90} (see also \cite{Stothers10,WilliamsSTOC12}),
they were analyzed individually through the concept of ``value'', 
a quantity that evaluates the number of square 
matrix products (and their size) that can be created from the form under consideration. 
This is nevertheless 
useless for estimating $\omega(1,1,k)$, except when $k=1$, 
because the value is intrinsically symmetric
(in particular, the values of $T_{112}$, $T_{121}$ and $T_{211}$ are identical),
while here we are precisely interested in breaking the symmetry in order to obtain 
bounds for rectangular matrix multiplication.
Instead, we will analyze the term 
\begin{equation}\label{eq_T112}
T_{112}^{\otimes a_{112}N}\otimes T_{121}^{\otimes a_{121}N}\otimes T_{211}^{\otimes a_{211}N}
\end{equation}
globally. This is the key new idea leading to our new bounds on $\alpha$ and more generally on 
$\omega(1,1,k)$ for any $k\neq 1$. Note that this difficulty was not present
in previous works on rectangular matrix multiplication \cite{Coppersmith97,Huang+98,Ke+08}: for 
simpler basic constructions such as $F_q$, all the smaller parts correspond to matrix products.
We will show that  $T_{112}$, $T_{121}$ and $T_{211}$ 
can be converted into a large number of 
objects called 
``$\Cc$-tensors'' in Strassen's terminology \cite{Strassen87}.
This will be done by relying on the graph-interpretation we introduced and showing how this conversion 
can be interpreted as finding large cliques in a graph (this is the main reason why we developed 
this graph-theoretic interpretation). 
While the fact that the form $T_{112}$ corresponds to a sum of 
$\Cc$-tensors was briefly mentioned in~\cite{Coppersmith+90},
and a proof sketched, 
we will need a complete analysis here. 
We will in particular rely on the fact that the $\Cc$-tensors obtained from $T_{112}$, 
$T_{121}$ and $T_{211}$ are not identical. 
The success of our approach comes from
the discovery that these $\Cc$-tensors are actually ``complementary'': while the $\Cc$-tensors obtained individually from 
$T_{112}$, $T_{121}$ and $T_{211}$ do not give any improvement for 
the exponent of rectangular matrix multiplication, their combination (i.e., the $\Cc$-tensors 
corresponding to the whole term (\ref{eq_T112})) does lead to improvements when analyzed globally
by the ``laser method'' developed by Strassen \cite{Strassen87}.
This will show that the form (\ref{eq_T112}), and thus the 
form (\ref{eq_T}) as well,
is isomorphic to a direct sum of  matrix products $\braket{m,m,m^k}$.

Finally, Sch\"onhage's asymptotic sum inequality will give an inequality, 
depending on the parameters~$a_{ijk}$,
that involves $\omega(1,1,k)$.
Our new upper bounds on $\omega(1,1,k)$ are 
obtained by optimizing these parameters. While 
this is done essentially though numerical calculations, 
the new lower bound on $\alpha$ requires 
a more careful analysis where the optimal values of all
but a few of the 	parameters are found analytically.

\paragraph{Higher powers.}
A natural question is whether our bounds can be improved by taking higher tensor powers of $F_q$
as the basic algorithm, i.e., taking $F_q^{\otimes r}$for $r>2$. As can be expected from the analysis for $r=2$ we have outlined above, the analysis is much more difficult than in the square case, for two main reasons. The first reason
is that the construction is not symmetric and thus more parameters 
have to be considered. This problem can be nevertheless addressed through a systematic framework 
similar to the one described in~\cite{WilliamsSTOC12} --- this is actually quite accessible without using a computer 
for $r=4$. 
The second, and more
fundamental, reason is that the analysis of the smaller parts is different from the square case since it does not use the concept of value. We solved this problem for $r=2$ by applying Strassen's laser method globally on the combination of $T_{112}$, $T_{121}$ and $T_{211}$, which is the key technical contribution of this paper.
For larger values of~$r$ the same approach can in principle be used, 
but other techniques seem to be necessary to convert these ideas into a systematic framework.

\paragraph{Organization of the paper.}
Section \ref{sec_prelim} describes formally the notions of trilinear forms
and Strassen's laser method. Section \ref{section_CW} presents in details
the basic construction $F_q\otimes F_q$ from \cite{Coppersmith+90} and some of its properties.
Section \ref{sec_graphtheory} describes the
graph-theoretic problems that will arise in the analysis of our trilinear forms.
Section \ref{sec_construction} describes our algorithm, and 
Section \ref{sec_analysis}  shows how to use this algorithm to
derive the optimization problem giving our new upper bounds on $\omega(1,1,k)$.
Finally, the optimization is done in Section \ref{sec_optimization}.

\section{Preliminaries}\label{sec_prelim}
In this section we present known results about trilinear forms, 
Strassen's laser method and Salem-Spencer sets
that we will use in this paper.
We refer to \cite{Burgisser+97} for an extensive treatment of the first two topics.

\subsection{Trilinear forms, degeneration of tensors and Sch\"onhage's asymptotic sum inequality}
It will be convenient for us to use an abstract approach and represent trilinear forms as tensors. 
Our presentation is independent from what was informally defined and stated in 
Subsection~\ref{sub_overview}, but describes essentially the same contents.
We thus encourage the reader who encounter these notions for the first time to refer  
at Subsection~\ref{sub_overview} for concrete illustrations.

We assume that $\field$ is an arbitrary field.

Let $U=\field^u$, $V=\field^v$ and $W=\field^w$ be three vector spaces over $\field$, where $u,v$
and $w$ are three positive integers.
A tensor $t$ of format $(u,v,w)$ is an element of $U\otimes V\otimes W=\field^{u\times v\times w}$,
where $\otimes$ denotes the tensor product.
If we fix bases $\{x_i\}$, $\{y_j\}$ and $\{z_k\}$ of $U$, $V$ and $W$, respectively,
then we can express $t$ as
$$
t=\sum_{ijk}t_{ijk}\:x_i\otimes y_j\otimes z_k
$$
for coefficients $t_{ijk}$ in $\field$. The tensor $t$ can then be represented by the 3-dimensional 
array $[t_{ijk}]$. We will often write $x_i\otimes y_j\otimes z_j$ simply as $x_iy_jz_k$. 

The tensor corresponding to the matrix multiplication of an $m\times n$ matrix by an $n\times p$
matrix
is the tensor of format $(m\times n,n\times p,m\times p)$ with coefficients 
$$
t_{ijk}=
\left\{
\begin{array}{ll}
1&\textrm{ if } i=(r,s), j=(s,t) \textrm{ and } k=(r,t) \textrm{ for some integers } (r,s,t)\in [m]\times [n]\times [p]\\
0&\textrm{ otherwise }.
\end{array}
\right.
$$
This tensor will be denoted by $\braket{m,n,p}$.

Another important example is the tensor $\sum_{\ell=1}^n x_\ell y_\ell z_\ell$ of format $(n,n,n)$.
This tensor is denoted $\braket{n}$ and corresponds to $n$ independent scalar products.

Let $\lambda$ be an indeterminate over $\field$ and consider the extension
$\field[\lambda]$ of $\field$, i.e., the set of all polynomials over $\field$ in $\lambda$. 
A triple of matrices $\alpha\in \field[\lambda]^{u'\times u}$, $\beta\in \field[\lambda]^{v'\times v}$ and
$\gamma\in \field[\lambda]^{w'\times w}$ transforms $t\in \field^{u\times v\times w}$ into the tensor 
 $(\alpha\otimes\beta\otimes\gamma)t\in \field[\lambda]^{u'\times v'\times w'}$ defined as
 $$
(\alpha\otimes\beta\otimes\gamma)t=\sum_{ijk}t_{ijk}\:\alpha(x_i)\otimes \beta(y_j)\otimes \gamma(z_j).
$$
This new tensor is called a restriction of $t$. Intuitively, the fact that a tensor $t'$ is a restriction 
of $t$ means that an algorithm computing $t$ can be converted into an algorithm computing $t'$
that uses the same amount of multiplications (i.e.,  an algorithm with essentially the same complexity).
We now give the definition of degeneration of tensors.
\begin{definition}\label{def_deg}
Let $t\in \field^{u\times v\times w}$ and $t'\in \field^{u'\times v'\times w'}$ be two tensors. 
We say that $t'$ is 
a degeneration of $t$, denoted $t'\degen t$, if there exists matrices 
$\alpha\in \field[\lambda]^{u'\times u}$, $\beta\in \field[\lambda]^{v'\times v}$ and
$\gamma\in \field[\lambda]^{w'\times w}$ such that
$$
\lambda^st'+\lambda^{s+1} t''=(\alpha\otimes\beta\otimes\gamma)t
$$
for some tensor $t''\in \field[\lambda]^{u'\times v'\times w'}$ and some nonnegative integer $s$.
\end{definition}
Intuitively, the fact that a tensor $t'$ is a degeneration of a tensor $t$ means that an 
algorithm computing~$t$ can be converted into an ``approximate 
algorithm'' computing $t'$ with essentially the same complexity.
The notion of degeneration can be used to define the notion of border rank. 
\begin{definition}
Let $t$ be a tensor. Then $\underline{R}(t)=\min\{{r\in\Nat\:|\: t\degen\braket{r}\}}$.
\end{definition}


Let $t\in U\otimes V\otimes W$ and $t'\in U'\otimes V'\otimes W'$ be two tensors.
We can naturally define the direct sum $t\oplus t'$, which is a tensor in 
$(U\oplus U')\otimes (V\oplus V')\otimes (W\oplus W')$, and the tensor product 
$t\otimes t'$, which is a tensor in 
$(U\otimes U')\otimes (V\otimes V')\otimes (W\otimes W')$. 
For any integer $c\ge 1$, we will denote the tensor $t\oplus\cdots\oplus t$ (with $c$ occurrences of $t$)
by $c\cdot t$ and 
the tensor $t\otimes\cdots \otimes t$ (with $c$ occurrences of $t$)
by $t^{\otimes c}$.
The degeneration of tensors has the following properties.


\begin{proposition}[Proposition 15.25 in \cite{Burgisser+97}]\label{prop_fact}
Let $t_1,t_1',t_2$ and $t_2'$ be four tensors.
Suppose that $t'_1\unlhd t_1$ and $t'_2\unlhd t_2$.
Then $t'_1\oplus t'_2\unlhd t_1\oplus t_2$ and $t'_1\otimes t'_2\unlhd t_1\otimes t_2$.
\end{proposition}

Sch\"onhage's asymptotic sum inequality \cite{Schonhage81} will be one of the main tools used to prove our bounds.
Its original statement is for estimating the exponent of square matrix multiplication, but it can be easily 
generalized to estimate the exponent of rectangular matrix multiplication as well. 
We will use the following form, which has been also used implicitly in \cite{Huang+98,Ke+08}.
A proof can be found in \cite{Lotti+83}.

\begin{theorem}[Sch\"onhage's asymptotic sum inequality] \label{theorem_schonhage}
Let $k$, $m$ and $c$ be three positive integers.
Let $t$ be a tensor such that $c\cdot \braket{m,m,m^k}\unlhd t$.
Then 
$$
c\cdot  m^{\omega(1,1,k)}\le \underline{R}(t).
$$
\end{theorem}
Theorem \ref{theorem_schonhage} states that, if the form $t$ 
can be degenerated (i.e., approximately converted, in the sense of Definition \ref{def_deg}) into a direct sum of $c$ forms, each being
isomorphic to $\braket{m,m,m^k}$, then the inequality 
$c\cdot  m^{w(1,1,k)}\le \underline{R}(t)$ holds. Note that this is a powerful technique, since
the concepts of degeneration and border rank refer to ``approximate algorithms'',
while $\omega(1,1,k)$ refers to the complexity of exact algorithms for rectangular matrix multiplication.

\subsection{Strassen's laser method and $\Cc$-tensors}
Strassen \cite{StrassenFOCS86} introduced in 1986 
a new approach, often referred as the laser method, to derive
upper bounds on the exponent of matrix multiplication. To the best of 
our knowledge, all the applications 
of this method have focused so far on square matrix multiplication, in which case
several simplifications can be done due to the symmetry of the problem. 
In this paper we will nevertheless need the full power of the laser method, and in particular 
the notion of $\Cc$-tensor introduced in \cite{StrassenFOCS86,Strassen87}, to 
derive our new bounds on the exponent of rectangular matrix multiplication. 
The exposition below will mainly follow~\cite{Burgisser+97}.

Let $t\in U\otimes V\otimes W$ be a tensor. 
Suppose that $U$, $V$ and $W$ decompose as direct sums of subspaces as follows:
$$
U=\bigoplus_{i\in S_U}U_i,\hspace{3mm}
V=\bigoplus_{j\in S_V}V_j,\hspace{3mm}
W=\bigoplus_{k\in S_W}W_k.
$$
Denote by $D$ this decomposition.
We say that $t$ is a $\Cc$-tensor with respect to $D$ if $t$ can be written as 
$$t=\sum_{(i,j,k)\in S_U\times S_V\times S_W} t_{ijk}$$ 
where each $t_{ijk}$ is a tensor in $U_i\otimes V_j\otimes W_k$.
The support of $t$ is defined as
$$
\supp_D(t)=\{(i,j,k)\in S_U\times S_V\times S_W\:|\: t_{ijk}\neq 0\},
$$
and the nonzero $t_{ijk}$'s are called the components of $t$. 
We will usually omit the reference to $D$ when there is no ambiguity 
or when the decomposition does not matter.

As a simple example, consider the complete decompositions of the spaces $U=\field^{m\times n}$,
$V=\field^{n\times p}$ and $W=\field^{m\times p}$ (i.e., their decomposition as direct sums of 
one-dimensional subspaces, each subspace being spanned by one element of their basis). 
With respect to this decomposition, the tensor of matrix multiplication $\braket{m,n,p}$
is a $\Cc$-tensor with support
$$
\suppc(\braket{m,n,p})=\{((r,s),(s,t),(r,t))\:|\:(r,s,t)\in[m]\times[n]\times[p]\}
$$
where each component is trivial (i.e., isomorphic to $\braket{1,1,1}$). In this paper the notation $\suppc(\braket{m,n,p})$
will always refer to the support of $\braket{m,n,p}$ with respect to this complete decomposition.



We now introduce the concept of combinatorial degeneration.
A subset $\Delta$ of $S_U\times S_V\times S_W$ is called diagonal if the three projections
$\Delta\to S_U$, $\Delta\to S_V$ and $\Delta\to S_W$ are injective.
Let $\Phi$ be a subset of $S_U\times S_V\times S_W$. A set $\Psi\subseteq \Phi$
is a combinatorial degeneration of $\Phi$ if there exists tree functions
$a\colon S_U\to\Int$, $b\colon S_V\to\Int$ and $c\colon S_W\to\Int$ such that 
\begin{itemize}
\item 
for all $(i,j,k)\in\Psi$, $a(i)+b(j)+c(k)=0$;
\item 
for all $(i,j,k)\in\Phi\backslash\Psi$, $a(i)+b(j)+c(k)>0$.
\end{itemize}

The most useful application of combinatorial degeneration will be 
the following result, which states that a sum,
over indices in a diagonal combinatorial degeneration of  $\supp_D(t)$,
of the components $t_{ijk}$ is direct.
\begin{proposition}[Proposition 15.30 in \cite{Burgisser+97}]\label{prop_deg}
Let $t$ be $\Cc$-tensor with support $\supp_D(t)$ and components~$t_{ijk}$.
Let $\Delta\subseteq \supp_D(t)$
be a combinatorial degeneration of $\supp_D(t)$ 
and assume that $\Delta$ is diagonal.
Then
$$
\bigoplus_{(i,j,k)\in\Delta}t_{ijk}\degen t.
$$
\end{proposition}

In this work we will construct $\Cc$-tensors where all the components 
are isomorphic to $\braket{m,m,m^k}$ for some values $m$ and $k$. 
Proposition \ref{prop_deg} then suggests
that a good bound on the exponent of rectangular matrix multiplication
can be derived, via Theorem \ref{theorem_schonhage}, if the support of the
$\Cc$-tensor contains a large diagonal combinatorial degeneration. 
When this support is isomorphic to $\suppc(\braket{e,h,\ell})$ for some positive
integers $e,h$ and $\ell$, 
a powerful
tool to construct large diagonal combinatorial degenerations is given by the following result 
by Strassen (Theorem 6.6 in \cite{Strassen87}), restated in our terminology.

\begin{proposition}[\cite{Strassen87}]\label{prop_Strassen}
Let $e_1, e_2$ and $e_3$ be three positive integers such that $e_1\le e_2\le e_3$.
For any permutation $\sigma$ of $\{1,2,3\}$, 
there exists a diagonal set $\Delta \subseteq \suppc(\braket{e_{\sigma(1)},e_{\sigma(2)},e_{\sigma(3)}})$  with
$$
|\Delta|=
\left\{
\begin{array}{ll}
e_1e_2-\floor{\frac{(e_1+e_2-e_3)^2}{4}}&\textrm{ if }e_1+e_2\ge e_3\\
e_1e_2&\textrm{ otherwise }
\end{array}
\right.
$$
that is a combinatorial degeneration of $\suppc(\braket{e_{\sigma(1)},e_{\sigma(2)},e_{\sigma(3)}})$.
In particular, $|\Delta|\ge \ceil{3e_1e_2/4}$.

\end{proposition}

Finally, we mention that the concept of $\Cc$-tensor is preserved by the tensor product.
We will just state this property for the restricted class of $\Cc$-tensors that we will encounter in this paper
(for which the precise decompositions do not matter),
and refer to \cite{Burgisser+97} or to Section 7 in \cite{Strassen87} for a complete treatment. 
\begin{proposition}
Let $t$ be a $\Cc$-tensor with support isomorphic to 
$\suppc(\braket{e,h,\ell})$ in which each component is isomorphic to $\braket{m,n,p}$.
Let $t'$ be a $\Cc$-tensor with support isomorphic to 
$\suppc(\braket{e',h',\ell'})$ in which each component is isomorphic to $\braket{m',n',p'}$.
Then $t\otimes t'$ is a $\Cc$-tensor with support isomorphic to 
$\suppc(\braket{ee',hh',\ell\ell'})$ in which each component is isomorphic to $\braket{mm',nn',pp'}$.
\end{proposition}

\subsection{Salem-Spencer sets}
Let $M$ be a positive integer and consider $\Int_{M}=\{0,1,\ldots,M-1\}$.
We say that a set $B\subseteq\Int_M$ has no length-3 arithmetic progression if,
for 
any three elements $b_1,b_2$ and $b_3$ in $B$, 
$$
b_1+b_2=2b_3 \bmod M \Longleftrightarrow b_1=b_2=b_3. 
$$
Salem and Spenser have shown the existence of very dense sets with no
length-3 arithmetic progression.
\begin{theorem}[\cite{Salem+42}]\label{th_Salem}
For any $\epsilon>0$, there exists an integer $M_\epsilon$
such that, for any integer $M>M_\epsilon$, there is 
a set $B\subseteq \Int_M$ of size $|B|>M^{1-\epsilon}$ with no length-3 
arithmetic progression.
\end{theorem}
We refer to these sets as Salem-Spencer sets. They can be constructed in time polynomial in $M$.
Note that this construction has been improved by Behrend \cite{Behrend46}, but the above statement
will be enough for our purpose. 

\section{Coppersmith-Winograd's construction}\label{section_CW}
In this section we describe the construction by Coppersmith and Winograd
\cite{Coppersmith+90}, which we will use as the basis of our algorithm, and several of its properties. 

We start with the simpler construction presented in Section 7 of \cite{Coppersmith+90}.
For any positive integer $q$, let us define the following trilinear form $F_q$, where $\lambda$
is an indeterminate over $\field$.
\begin{align*}
F_{q}=&\sum_{i=1}^q \lambda^{-2}(x_0+\lambda x_i)(y_0+\lambda y_i)(z_0+\lambda z_i)\:-\\
&\lambda^{-3}(x_0+\lambda^2\sum_{i=1}^q x_i)(y_0+\lambda^2\sum_{i=1}^q y_i)(z_0+\lambda^2\sum_{i=1}^q z_i)\:+\\
&(\lambda^{-3}-q\lambda^{-2})(x_0+\lambda^3x_{q+1})(y_0+\lambda^3y_{q+1})(z_0+\lambda^3z_{q+1})
\end{align*}
In this trilinear form the $x$-variables are $x_0,x_1,\ldots, x_{q+1}$. Similarly, the number of  
$y$-variables is $(q+2)$ and the number of  $z$-variables is $(q+2)$ as well. 
Define the form
$$
F'_q=\sum_{i=1}^q (x_0y_iz_i+x_iy_0z_i+x_iy_iz_0)+x_0y_0z_{q+1}+x_0y_{q+1}z_{0}+x_{q+1}y_0z_0.
$$
It is easy to check that the form $F_q$ can be written as 
$F_q=F'_q+\lambda\cdot F''_q$, where $F''_q$ is a polynomial in $\lambda$ and in the $x$-variables, $y$-variables and $z$-variables. In the language of Section \ref{sec_prelim}, this means that $F'_q\degen F_q$ and, informally, this means
that an algorithm computing $F_q$ can be converted into an algorithm computing $F'_q$ with the same complexity.
Note that $\underline{R}(F_q)\le q+2$ since, by definition, 
$F_q$ is the sum of $q+2$ products.

A more complex construction is proposed in Section 8 of \cite{Coppersmith+90}.
It
is obtained by taking the tensor product of $F_q$ by itself. 
By Proposition \ref{prop_fact} we know that $F'_q\otimes F'_q\degen F_q\otimes F_q$.
Consider the tensor product of $F'_q$ by itself:
\begin{eqnarray*}
F'_q\otimes F'_q&=&
T_{004}+T_{040}+T_{400}+
T_{013}+T_{031}+T_{103}+T_{130}+T_{301}+T_{310}+\\
&&T_{022}+T_{202}+T_{220}+
T_{112}+T_{121}+T_{211}
\end{eqnarray*}
where 
\begin{eqnarray*}
T_{004}&=&x_{0,0}^{0}y_{0,0}^{0}z_{q+1,q+1}^{4}\\
T_{013}&=&\sum_{i=1}^q x_{0,0}^{0}y_{i,0}^{1}z_{i,q+1}^{3}+\sum_{k=1}^q x_{0,0}^{0}y_{0,k}^{1}z_{q+1,k}^{3}\\
T_{022}&=&x_{0,0}^{0}y_{q+1,0}^{2}z_{0,q+1}^{2}+x_{0,0}^{0}y_{0,q+1}^{2}z_{q+1,0}^{2}+\sum_{i,k=1}^q x_{0,0}^{0}y_{i,k}^{2}z_{i,k}^{2}\\
T_{112}&=&\sum_{i=1}^q x_{i,0}^{1}y_{i,0}^{1}z_{0,q+1}^{2}+\sum_{k=1}^q x_{0,k}^{1}y_{0,k}^{1}z_{q+1,0}^{2}+\sum_{i,k=1}^q x_{i,0}^{1}y_{0,k}^{1}z_{i,k}^{2}+\sum_{i,k=1}^q x_{0,k}^{1}y_{i,0}^{1}z_{i,k}^{2}
\end{eqnarray*}
and the other eleven terms are obtained by permuting the indexes of the $x$-variables, the $y$-variables and $z$-variables in the above expressions (e.g., $T_{040}=x_{0,0}^{0}y_{q+1,q+1}^{4}z_{0,0}^{0}$ and $T_{400}=x_{q+1,q+1}^{4}y_{0,0}^{0}z_{0,0}^{0}$). Let us describe in more details the notations used here. The number of $x$-variables is $(q+2)^2$. They are indexed as $x_{i,k}$, for $i,k\in \{0,1,\ldots, q+1\}$. The superscript is assigned in the following
way:
the variable $x_{0,0}$ has superscript $0$, 
the variables in $\{x_{i,0},x_{0,k}\}_{1\le i,j\le q}$ have superscript $1$,
the variables in $\{x_{q+1,0},x_{i,k},x_{0,q+1}\}_{1\le i,j\le q}$ have superscript $2$, 
the variables in $\{x_{q+1,k},x_{i,q+1}\}_{1\le i,j\le q}$ have superscript~$3$ 
and the variable $x_{q+1,q+1}$ has superscript $4$. 
Note that the superscript is completely determined by the subscript. 
Similarly, the number of $y$-variables is $(q+2)^2$, and 
the number of $z$-variables is $(q+2)^2$ as well. The 
$y$-variables and the $z$-variables are assigned subscripts and superscripts 
exactly as for the $x$-variables. Observe that any term $xyz$ that appears in 
$T_{ijk}$ is such that
$x$
has superscript~$i$, $y$
has superscript $j$ and $z$ has
superscript~$k$.

We thus obtain 
$$
\sum_{\begin{subarray}{l} 0\le i,j,k\le 4 \\ i+j+k=4 \end{subarray}}T_{ijk} \degen F_q\otimes F_q. 
$$
Moreover we know that $\underline{R}(F_q\otimes F_q)\le (q+2)^2$, since $F_q\otimes F_q$ can be 
written using $(q+2)^2$ multiplications.

We will later need to analyze all the forms $T_{ijk}$. It happens,
as observed in \cite{Coppersmith+90}, that most of
these forms (all the forms except $T_{112}$, $T_{121}$ and $T_{211}$)
can be analyzed in a straightforward way, since they are isomorphic to the following matrix products:
\begin{eqnarray*}
T_{004}\cong T_{040}\cong T_{400}&\cong&\langle 1,1,1\rangle\\
T_{013}\cong T_{031}&\cong&\langle 1,1,2q\rangle\\
T_{103}\cong T_{301}&\cong & \langle 2q,1,1\rangle\\
T_{130}\cong T_{310}&\cong & \langle 1,2q,1\rangle\\
T_{022}&\cong &\langle 1,1,q^2+2\rangle\\
T_{202}&\cong &\langle q^2+2,1,1\rangle\\
T_{220}&\cong &\langle 1,q^2+2,1\rangle .
\end{eqnarray*}
This can be seen from the definition of the trilinear form (or the tensor) corresponding 
to matrix multiplication described in Section \ref{sec_prelim}. For example, the form 
$T_{013}$ is isomorphic to the tensor 
$\sum_{\ell=1}^{2q}x_0y_\ell z_\ell=\langle 1,1,2q\rangle$, which 
represents
the product of a $1\times 1$ matrix (a scalar) by a $1\times 2q$
matrix (a row). 

\section{Graph-Theoretic Problems}\label{sec_graphtheory}
In this section we describe and solve several graph-theoretic problems
that will arise in the analysis of our trilinear forms. While the presentation
given here is independent from the remaining of the paper, the reader may prefer
to read Section \ref{sec_construction} before going through this section.

\subsection{Problem setting}\label{sub_graphtheorysetting}

Let $\tau$ be a fixed positive integer.
Let $N$ be a large integer and  
define the set 
$$\Lambda= \left\{(I,J,K)\in [\tau]^N\times [\tau]^N\times [\tau]^N \:|\: I_\ell+J_\ell+K_\ell=\tau \textrm{ for all } \ell\in\{1,\ldots,N\}\right\}.$$
Define the three coordinate functions $f_1,f_2,f_3\colon [\tau]^N\times [\tau]^N\times [\tau]^N\to [\tau]^N$ as follows.
\begin{eqnarray*}
f_1((I,J,K))&=&I\\
f_2((I,J,K))&=&J\\ 
f_3((I,J,K))&=&K
\end{eqnarray*}

From the definition of $\Lambda$, two distinct elements in $\Lambda$ cannot agree on more than 
one coordinate. Since this simple observation will be crucial in our analysis, we state it explicitly as 
follows.
\begin{fact}\label{fact}
Let $u$ and $v$ be two elements in $\Lambda$. If $f_i(u)=f_i(v)$ for more than one index
$i\in\{1,2,3\}$, then $u=v$.
\end{fact}

Let $U$ be a subset of $\Lambda$ such that there exist integers $\Nn_1,\Nn_2$ and $\Nn_3$ for which
the following property holds: for any $I\in [\tau]^N$, 
\begin{eqnarray*}
|\{u\in U \:|\: f_1(u)=I\}|&\in&\{0,\Nn_1\}\\
|\{u\in U \:|\: f_2(u)=I\}|&\in&\{0,\Nn_2\}\\
|\{u\in U \:|\: f_3(u)=I\}|&\in&\{0,\Nn_3\}.
\end{eqnarray*}
This means that, for any $i\in\{1,2,3\}$ and any $I\in [\tau]^N$,
the size of the set $f_i(I)^{-1}$ is either $0$ or $\Nn_i$.
Let us write $|f_1(U)|=T_1$, $|f_2(U)|=T_2$ and $|f_3(U)|=T_3$.
It is easy to see that 
$$
|U|=T_1\Nn_1=T_2\Nn_2=T_3\Nn_3.
$$
We are interested in stating asymptotic results holding when $N$ goes to infinity.
Through this section the $\Nn_i$'s and the $T_i$'s 
will be strictly increasing functions of $N$. 

Let $G$ be the (simple and undirected) graph with vertex set $U$ in which
two distinct vertices~$u$ and~$v$ are connected if and only there exists
one
index $i\in\{1,2,3\}$ such that $f_i(u)=f_i(v)$.
%
The goal will be to modify the graph $G$ to obtain a subgraph satisfying some specific properties. 
The only modification we allow is to remove all the vertices with a given sequence at a given position: 
given a sequence $I\in [\tau]^N$ and a position $s\in\{1,2,3\}$, remove all the vertices $u$ (if any) such that $f_i(u)=I$.
We call such an operation a removal operation.
The reason why only such removal operations are allowed is that they will correspond, when considering trilinear forms, 
to setting to zero some variables, which is one of the only operations that can be performed on trilinear forms.

While not stated explicitly in graph-theoretic terms, the key technical result by Coppersmith and Winograd \cite{Coppersmith+90} is a method to convert, 
when $\Nn_1=\Nn_2=\Nn_3$,
the graph $G$ into an edgeless graph that still contains a non-negligible fraction of the vertices.
We state this result in the following theorem.
\begin{theorem}[\cite{Coppersmith+90}]\label{th_CW}
Suppose that $\Nn_1=\Nn_2=\Nn_3$. Then, for any constant $\epsilon>0$, the graph $G$ can be converted, with only removal operations, into an edgeless graph with 
$\Omega\left(\frac{T_1}{\Nn_1^{\epsilon}}\right)$ vertices.
\end{theorem}

In this section we will give several generalizations of this result. 
In Subsection \ref{sub_stat} we state our generalizations. 
Then, in Subsections \ref{sub_weights}--\ref{sub_second2},
we describe the algorithms and prove the results.
We stress that, in this section as in \cite{Coppersmith+90}, the number of removal operations 
(i.e., the time complexity of the algorithms) is irrelevant. This is 
because, in the applications of these results to matrix multiplication,
the parameter $N$ will be treated as a large constant independent 
of the size of the matrices considered.

\subsection{Statement of our results}\label{sub_stat}

The generalization we consider assume the existence of a known set $U^\ast\subseteq U$ such that 
\begin{itemize}
\item
$|f_i(U^\ast)|=T_i$ for each $i\in\{1,2,3\}$;
\item
there exist integers $\Nn^\ast_1,\Nn^\ast_2$ and $\Nn^\ast_3$ such that
$$|\{u\in U \:|\: f_i(u)=I\}|=\Nn_i \Leftrightarrow |\{u\in U^\ast \:|\: f_i(u)=I\}|=\Nn_i^\ast$$
for each $I\in [\tau]^N$ and each $i\in\{1,2,3\}$.
\end{itemize}
Note that we have necessarily $T_1\Nn^\ast_1=T_1\Nn^\ast_2=T_1\Nn^\ast_2$.

The first problem considered is again to convert the graph $G$ into an edgeless subgraph 
that contains a non-negligible fraction of the vertices using only removal operations, 
but we additionally require that all the remaining vertices are in $U^\ast$.
Our first result is the following theorem.
\begin{theorem}\label{th_1}
For any constant $\epsilon>0$, the graph $G$ can be converted, with only removal operations, into an edgeless graph with 
$$\Omega\left(\frac{T_1\Nn^\ast_1}{(\Nn_1+\Nn_2+\Nn_3)^{1+\epsilon}}\right)$$ vertices, all of them being in $U^\ast$.
\end{theorem}
Theorem \ref{th_CW} is a special case of Theorem \ref{th_1}, for $U^\ast=U$.
The case $\Nn_1=\Nn_2=\Nn_3$ has been proved implicitly by Stothers \cite{Stothers10}
and Vassilevska Williams \cite{WilliamsSTOC12}.

Our second problem deals with another kind of conversion.
Remember that, from the definition of the graph $G$ 
and Fact \ref{fact}, for any edge connecting
two vertices $u$ and $v$ in $G$ there exists
exactly one index $i\in\{1,2,3\}$ such that $f_i(u)=f_i(v)$.
Let us define the concept of 1-clique as follows.

\begin{definition}
A 1-clique in the graph $G$ is  
a set $U'\subseteq U$ for which
there exists a sequence $I\in [\tau^N]$ such that $f_{1}(u)=I$ for all $u\in U'$.
The size of the 1-clique is $|U'|$.
\end{definition}

Our second result shows how to convert, using only removal operations,
the graph $G$ into a graph with vertices in $U^\ast$ that is the disjoint
union of many large 1-cliques. The formal statement follows. 
\begin{theorem}\label{th_2}
Suppose that $\Nn_1\ge \Nn_2\ge \Nn_3$.
Assume that
$\frac{\Nn_2T_1}{\Nn_1^\ast}+\frac{\Nn_2}{T_1}<\frac{1}{1024}$.
Then, for any constant $\epsilon>0$, 
the graph $G$ can be converted, with only removal operations, into a graph satisfying the following conditions:
\begin{itemize}
\item 
all the vertices of the graph are in $U^\ast$;
\item
each connected component is a 1-clique (i.e., the graph is a disjoint union of 1-cliques);
\item
among these 1-cliques, there are
$\Omega\left(\frac{T_1}{\Nn_2^{\epsilon}}\right)$ 1-cliques that have size 
$\Omega\left(\frac{\Nn_1^\ast}{\Nn_2}\right)$.
\end{itemize}
\end{theorem}
Note that, when only removal operations are allowed, the graph obtained 
is necessary a subgraph of~$G$ induced be a subset of its vertices.
Theorem \ref{th_2} thus states that there exist 
1-cliques $U_r\subseteq U^\ast$ such 
that 
the graph 
obtained after the removal operations is the subgraph of $G$ induced by $\cup_r U_r$, 
at least $\Omega\left(T_1/ \Nn_2^{\epsilon}\right)$ of these
$U_r$'s have size $\Omega\left(\Nn_1^\ast/ \Nn_2\right)$, and
for any $r\neq r'$ there is no edge with one extremity in~$U_r$ and the other extremity
in~$U_{r'}$.
We mention that the constant $1/1024$ in the assumption of Theorem~\ref{th_2} 
is chosen only for concreteness. The same theorem actually holds even for  
weaker conditions on $\Nn^\ast_1, T_1$ and $\Nn_2$,  
but this simpler version will be sufficient for our purpose.

\subsection{Choice of the weight functions}\label{sub_weights}
Let $M$ be a large prime number that will be chosen later.
We take a Salem-Spencer set $\Gamma$ with $|\Gamma|\ge M^{1-\epsilon}$, which existence is 
guaranteed by Theorem \ref{th_Salem}.
Similarly to $\cite{Coppersmith+90}$, we take $N+2$ integers $\omega_0$, $\omega_1$,$\ldots$, $\omega_{N+1}$ uniformly at random in $\Int_M=\{0,1,\ldots,M-1\}$ 
and define three hash functions $b_1,b_2,b_3\colon [\tau]^N\to \Int_M$ as follows.
\begin{eqnarray*}
b_1(I)&=&\omega_0+\sum_{j=1}^N I_j \omega_j \bmod M\\
b_2(I)&=&\omega_{N+1}+\sum_{j=1}^N I_j \omega_j \bmod M\\
b_3(I)&=&\frac{1}{2}\times\left(\omega_0+\omega_{N+1}+\sum_{j=1}^N (\tau-I_j) \omega_j\right) \bmod M
\end{eqnarray*}

A property of these functions is that, for a fixed sequence $I\in[\tau]^N$, the value $b_i(I)$
is uniformly distributed in $\Int_M$, for each $i\in\{1,2,3\}$. Note that the term $\omega_0$ is not used in~\cite{Coppersmith+90},
but we introduce it to obtain a uniform distribution even if $I$ is the all-zero sequence.

Let us introduce the notion of a compatible vertex.
\begin{definition}
A vertex $u\in{U}$ is compatible if 
$b_i(f_i(u))\in \Gamma$ for each $i\in\{1,2,3\}$.
Let $V$ be the set of vertices in $U$ that are compatible,
and denote $V^\ast=V\cap U^\ast$. 
\end{definition}
We stress that the definition of compatibility 
(and thus the definitions of $V$ and $V^\ast$ too)
depend
of the choice of $\Gamma$ and of the weights $\omega_i$.
Using the fact that $\Gamma$ is a Salem-Spencer set  we can give the following simple but very useful characterization of compatible vertices.
\begin{lemma}\label{lemma_comp}
Let $i\in\{1,2,3\}$ and $i'\in\{1,2,3\}$ be any two distinct indexes.
Then a vertex $u\in U$ is compatible if and only if $b_i(f_i(u))=b_{i'}(f_{i'}(u))=b$ for some $b\in\Gamma$.
\end{lemma}
\begin{proof}
Let us take a vertex $u\in U$, and write $f_1(u)=I$, $f_2(u)=J$ and $f_3(u)=K$.

Note that, since $I_\ell+J_\ell+K_\ell=\tau$ for any index $\ell\in\{1,\ldots,N\}$, the equality
$b_1(I)+b_2(J)-2 b_3(K)=0\bmod M$
always holds (i.e., holds for all choices of the weights $\omega_j$).

Suppose that $b_i(f_i(u))=b_{i'}(f_{i'}(u))=b$ for some $b\in\Gamma$, where $i$ and $i'$
are distinct.
For instance, suppose that $i=1$ and $i'=2$. Then, since $M$ is a 
prime, the above property implies that $b_3(K)=b$, which means that $u$ is compatible.
The same conclusion is true for the other choices of $i$ and $i'$.

Now suppose that $u$ is compatible.
From the definition of a Spencer-Salem set, we can conclude that there exists 
an element $b\in \Gamma$ such that $b_1(I)=b_2(J)=b_3(K)=b$.
\end{proof}

\subsection{The first pruning}\label{sub_first}
Similarly to \cite{Coppersmith+90}, 
the first pruning simply eliminates all the nodes $u\in U$ that are not compatible.
Note that this can be done using removal operations: for each $i\in\{1,2,3\}$, we remove all the vertices 
$w$ such that $f_i(w)\notin b_i^{-1}(\Gamma)$. 
The vertices remaining are precisely those in $V$.
Among those remaining vertices, the vertices in $U^\ast$
are precisely those in $V^\ast=V\cap U^\ast$. 

We now evaluate the expectation of $|V^\ast|$.
The proof is very similar to what was shown in \cite{Coppersmith+90}
for the case $U=U^\ast$, and to what was shown in 
\cite{Stothers10,WilliamsSTOC12} for $\Nn^\ast_1=\Nn^\ast_2=\Nn^\ast_3$.
\begin{lemma}\label{lemma_Sast}
$\E[|V^\ast|]=\frac{T_1\Nn_1^\ast|\Gamma|}{M^2}$.
\end{lemma}
\begin{proof}
We use Lemma \ref{lemma_comp} with $i=1$ and $i'=2$. Remember that $|U^\ast|=T_1\Nn_1^\ast$.
For each vertex $u\in U^\ast$ and each value $b\in\Int_M$, 
the probability that $b_1(f_1(u))=b_2(f_2(u))=b$ is $1/M^2$. Note that the two events 
$b_1(f_1(u))=b$ and $b_2(f_2(u))=b$ are independent even when $f_1(u)=f_2(u)$ 
due to the terms $\omega_0$ and $\omega_{N+1}$ in the hash functions.
\end{proof}

Let $E$ be the edge set of the subgraph of $G$ induced by $V$: it consists of all
edges connecting 
two distinct vertices $u$ and $v$ in $V$ such that $f_i(u)= f_i(v)$ for some $i\in\{1,2,3\}$. 
Let $E'\subseteq E$ be the subset of edges in $E$ with (at least) 
one extremity in $V^\ast$.
Let $E''\subseteq E'$ be the subset of edges in $E'$ connecting two vertices $u$ and $v$
such that $f_1(u)\neq f_1(v)$ (which means that either $f_2(u)= f_2(v)$ or $f_3(u)= f_3(v)$).
The following lemma gives upper bounds on the expectations of $|E'|$ and $|E''|$. 
The proof can be considered as a generalization of similar statements in 
\cite{Coppersmith+90,Stothers10,WilliamsSTOC12}.

\begin{lemma}\label{lemma_W}
$\E[|E'|]\le \frac{T_1\Nn_1^\ast (\Nn_1+\Nn_2+\Nn_3-3)|\Gamma|}{M^3}$ and $\:\E[|E''|]\le \frac{T_1\Nn_1^\ast (\Nn_2+\Nn_3-2)|\Gamma|}{M^3}$.
\end{lemma}
\begin{proof}
We show that, for any index $i\in \{1,2,3\}$, the expected number of ordered pairs $(u,v)$ with 
$u\in V^\ast$ and $v\in V\backslash\{u\}$ such that $f_i(u)=f_i(v)$ is exactly 
$$\frac{T_1\Nn_1^\ast (\Nn_i-1)|\Gamma|}{M^3}.$$
The expectation of the number of edges (i.e., unordered pairs) is necessarily smaller.

The factor $T_1\Nn_1^\ast$ counts the number of vertices in $U^\ast$.
For any vertex $u\in U^\ast$,  
there are exactly $\Nn_i-1$ vertices $v\in U\backslash\{u\}$ such that $f_i(v)=f_i(u)$. 

Let $u$ be a vertex in $U^\ast$ and $v$ be a vertex in $U\backslash\{u\}$
such that $f_i(v)=f_i(u)$.
Take another index
$i'\in\{1,2,3\}\backslash\{i\}$ arbitrarily.
From Lemma \ref{lemma_comp}, both $u$ and $v$ are in $V$ if and only if 
$$b_i(f_i(u))=b_{i'}(f_{i'}(u))=b_{i'}(f_{i'}(v))=b$$
for some element $b\in \Gamma$.
This happens with probability $|\Gamma|/M^3$ since the random variables
$b_i(f_i(u))$, $b_{i'}(f_{i'}(u))$, and $b_{i'}(f_{i'}(v))$ are mutually independent.
From the linearity of the expectation, the expected number of ordered pairs is 
$T_1\Nn_1^\ast(\Nn_i-1)$ times this probability.
\end{proof}

Lemmas \ref{lemma_Sast} and \ref{lemma_W} 
focused on expected values and
were proven, essentially, by the linearity of the expectation. 
This was sufficient for the applications to square matrix multiplication
presented in~\cite{Coppersmith+90,Stothers10,WilliamsSTOC12}, and 
this will be sufficient for proving Theorem \ref{th_1} as well.
Since, in order to prove Theorem \ref{th_2}, we will need 
a more precise analysis of the behavior of the random variables
considered, we now prove the following lemma, which gives a lower bound on the 
probability that the subgraph induced by~$V^\ast$ has many large $1$-cliques.
\begin{lemma}\label{lemma_R}
With probability (on the choice of the weights $\omega_i$) at least $1-\frac{4MT_1}{\Nn_1^\ast}-\frac{4M}{T_1|\Gamma|}$, 
there exists
a 
set $R\subseteq [\tau]^N$  satisfying the following two 
conditions:
\begin{itemize}
\item
$|R|\ge \frac{T_1|\Gamma|}{2M}$;
\item
for any
$I\in R$, there exist at least $\frac{\Nn_1^\ast}{2M}$ vertices $u\in V^\ast$ such that $f_1(u)=I$.
\end{itemize}
\end{lemma}
\begin{proof}
Let us consider the set $f_1(U^\ast)=\{f_1(u)\:|\: u\in U^\ast\}$. For each element $I\in f_1(U^\ast)$, define the
set $$S_I=\{u\in U^\ast\:|\: f_1(u)=I\}.$$ Note that $|f_1(U^\ast)|=T_1$, and $|S_I|=\Nn_1^\ast$ for each $I\in f_1(U^\ast)$.

Fix an element $I\in f_1(U^\ast)$ and define
$$X_I=|\{u\in S_I\:|\: b_2(f_2(u))=b_1(I)\}|.$$ This random variable
represents the number of vertices from $S_I$ that are
mapped by $b_2\circ f_2$ into $b_1(I)$. 
For any vertex $u=(I,J,K)\in S_I$, we have the equivalence
$$
b_2(f_2(u))=b_1(I)\Longleftrightarrow \sum_{i=1}^N(J_i-I_i)\omega_{i}=\omega_{0}-\omega_{N+1}.
$$
Thus the probability of the event $b_2(f_2(u))=b_1(I)$ is $1/M$.
Thus $\E[X_I]=\frac{\Nn_1^\ast}{M}$.
Note that, for any 
two distinct elements $u$ and $v$ in $S_I$, the two events $b_2(f_2(u))=b_1(I)$ and 
$b_2(f_2(v))=b_1(I)$ are independent. 
From this pairwise independence, $\var[X_I]=\frac{\Nn_1^\ast}{M}(1-\frac{1}{M})$.
This gives
\begin{eqnarray*}
\Pr\left[\left|X_I-\frac{\Nn_1^\ast}{M}\right|\ge \frac{\Nn_1^\ast}{2M}\right]\le \Pr\left[\left|X_I-\frac{\Nn_1^\ast}{M}\right|\ge 
 \sqrt{\frac{\Nn_1^\ast}{4M}}\cdot \sqrt{\var(X_I)}\right]\le \frac{4M}{\Nn_1^\ast},
\end{eqnarray*}
where the second inequality is obtained by Chebyshev's inequality.


By the union bound we can conclude that 
$$
\Pr\left[|X_I|\ge \frac{\Nn_1^\ast}{2M} \textrm{ for all }I\in f_1(U^\ast)\right]\ge 1-\frac{4MT_1}{\Nn_1^\ast}.
$$

Define $R=\{I\in f_1(U^\ast) \:|\: b_1(I)\in \Gamma\}$.
For each element $I\in f_1(U^\ast)$, the variable $b_1(I)$ is distributed uniformly at random 
in $\Int_M$. Moreover, the variables $b_1(I)$'s are pairwise independent. 
Thus, by Chebyshev's inequality, we obtain
$$
\Pr\left[\left|R\right|\ge \frac{T_1|\Gamma|}{2M}\right]\ge 1-\frac{4M}{T_1|\Gamma|}.
$$

From the union bound we can conclude that, with probability at least $1-\frac{4MT_1}{\Nn_1^\ast}-\frac{4M}{T_1|\Gamma|}$,
the inequality $\left|R\right|\ge T_1|\Gamma|/(2M)$ holds and simultaneously,
for each element $I\in R$, 
there exist at least $\Nn_1^\ast/(2M)$ vertices $u\in U^\ast$ such that $f_1(u)=I$
and $b_2(f_2(u))=b_1(I)$.
These vertices are in $V^\ast$ by Lemma~\ref{lemma_comp}.
%
%
\end{proof}

\subsection{The second pruning: first version and proof of Theorem \ref{th_1}}\label{sub_second1}
In this subsection $M$ is an arbitrary prime number such that $2(\Nn_1+\Nn_2+\Nn_3)<M<4(\Nn_1+\Nn_2+\Nn_3)$.

The first pruning has transformed the graph $G$ into the subgraph induced by $V$.
The second pruning, similarly to \cite{Coppersmith+90}, will further modify this subgraph by removing vertices 
in order to obtain a subgraph consisting of isolated vertices from $V^\ast$ (i.e., an edgeless graph).
This is done by constructing greedily a set $L\subseteq V^\ast$ of isolated vertices. Initially
$L=\emptyset$ and, at each iteration, either one remaining vertex in $V^\ast\backslash L$ will be 
added to~$L$ or several vertices in $V$ will be removed. This will be repeated until there is no
remaining vertex in $V^\ast\backslash L$.
Finally, all the remaining vertices not in $L$ will be removed.
The detailed procedure is described in Figure \ref{fig_alg1}, where $\overline{V}$ 
represents the set of remaining vertices (initially $\overline{V}=V$).
Note that the procedure slightly differs from what was done in
\cite{Coppersmith+90,Stothers10,WilliamsSTOC12} since
we need to take in consideration the asymmetry of the problem.

\begin{figure}
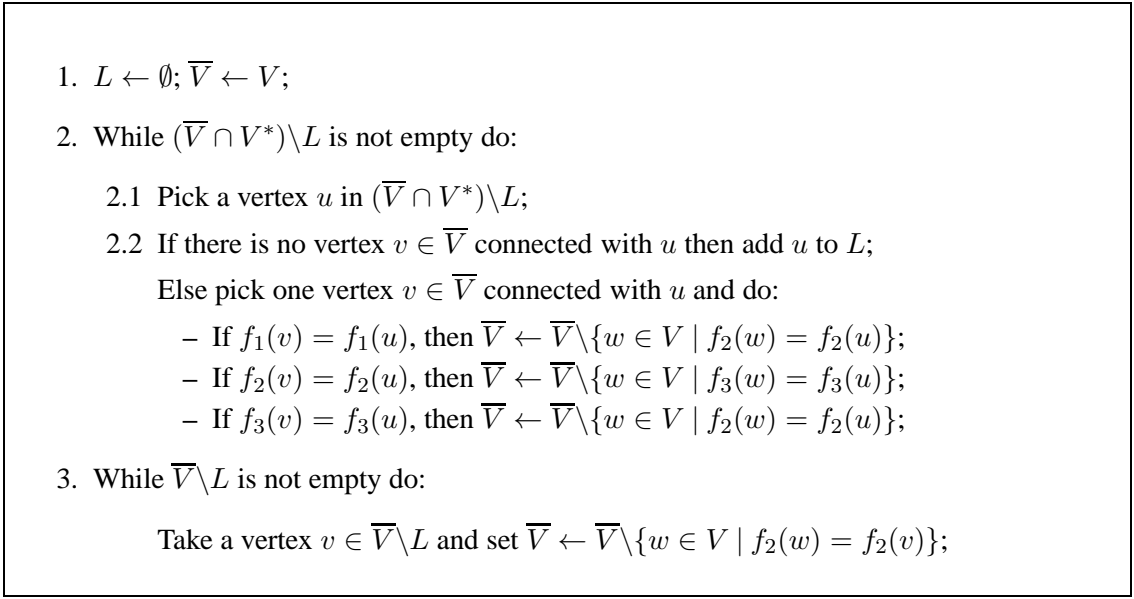

\begin{center}
    \begin{tabular}[h!]{|p{0.9\textwidth}|}
    \hline
\begin{enumerate}    
\item
$L\gets\emptyset$; $\overline{V}\gets V$;  
\item
While $(\overline{V}\cap V^\ast)\backslash L$ is not empty do: 
\begin{itemize}
\item[2.1]
Pick a vertex $u$ in $(\overline{V}\cap V^\ast)\backslash L$;
\item[2.2]
If there is no vertex $v\in \overline{V}$ connected with $u$  then add $u$ to $L$;
\item[]
Else pick one vertex $v\in \overline{V}$ connected with $u$ and do:
\begin{itemize}
\item
If $f_1(v)=f_1(u)$, then $\overline{V}\gets \overline{V}\backslash \{w\in V\:|\:f_2(w)=f_2(u)\}$;
\item
If $f_2(v)=f_2(u)$, then $\overline{V}\gets \overline{V}\backslash \{w\in V\:|\:f_3(w)=f_3(u)\}$;
\item
If $f_3(v)=f_3(u)$, then $\overline{V}\gets \overline{V}\backslash \{w\in V\:|\:f_2(w)=f_2(u)\}$;
\end{itemize}
\end{itemize}
\item
While $\overline{V}\backslash L$ is not empty do: 
\begin{itemize}
\item[]
Take a vertex $v\in \overline{V}\backslash L$ and set $\overline{V}\gets \overline{V}\backslash \{w\in V\:|\:f_2(w)=f_2(v)\}$;
\end{itemize}
\end{enumerate}
        \\\hline
    \end{tabular}\caption{The second pruning (first version) \label{fig_alg1}}
\end{center}
\end{figure}

Let $\overline{V_f}$ denote the contents of $\overline{V}$ at the end of the procedure.
The following proposition, shown using the same ideas as in \cite{Coppersmith+90},
 shows that what we obtain is a large set of isolated vertices from $V^\ast$.

\begin{proposition}\label{prop_secondpruningv1}
The subgraph of $G$ induced by $\overline{V_f}$ is an edgeless graph.
Moreover, $\overline{V_f}\subseteq V^\ast$ and the expectation of $|\overline{V_f}|$ (over the choices of $\omega_j$) is 
$$\Omega\left(\frac{T_1 \Nn_1^\ast}{(\Nn_1+\Nn_2+\Nn_3)^{1+\epsilon}}\right).$$
\end{proposition}
\begin{proof}
Let $L_f$ denote the contents of $L$ at the end of the procedure.
First observe that $\overline{V_f}\subseteq L_f$, due to Step 3.
Note that any vertex added to $L$ cannot be later removed from $\overline{V}$,
since it has no neighbor.
Thus $L_f\subseteq\overline{V_f}$, and we conclude that  
$\overline{V_f}= L_f$.
This shows in particular that $\overline{V_f}\subseteq V^\ast$.
Moreover, since each vertex in $L$ has no neighbor, the subgraph induced by $\overline{V_f}$ is edgeless. 

To prove the second part, we will show an upper bound on the number of vertices in $V^\ast$ removed from $\overline{V}$
during the loop of Step~2. The bound will be obtained by considering the number of edges from $E'$ remaining
in the subgraph induced by $\overline{V}$.

Let us consider what happens during Step 2.2.
Let $u\in(\overline{V}\cap V^\ast)\backslash L$ be the vertex currently examined.
Suppose that another vertex in $\overline{V}$ sharing one index with $u$ is found. 
For example, suppose that we find another vertex $v\in\overline{V}$ with $f_1(v)=f_1(u)$.
Let $S=\{w\in \overline{V}\cap V^\ast\:|\:f_2(w)=f_2(u)\}$ be the set of vertices in $V^\ast$ eliminated by the consequent removal operation. 
Observe that this removal operation will eliminate at least ${|S| \choose 2}+1$ new edges from $E'$: 
the edges between two vertices in $S$, and the edge connecting $u$ and~$v$. 
Since ${|S| \choose 2}+1\ge |S|$, the number of vertices in $V^\ast$ removed during one execution of 
Step 2.2 is at most the number of edges eliminated from $E'$. 

The total number of vertices from $V^\ast$ that are removed by the procedure during the loop of Step~2 
is thus at most $|E'|$, which means that $|\overline{V_f}\cap V^\ast|\ge |V^\ast|-|E'|$. 
Since $\overline{V}\subseteq V^\ast$,
Lemmas \ref{lemma_Sast} and \ref{lemma_W} imply that
the expected number of vertices in 
$\overline{V_f}$ is at least
\begin{eqnarray*}
\frac{T_1\Nn_1^\ast|\Gamma|}{M^2}\left(1 - \frac{(\Nn_1+\Nn_2+\Nn_3-3)}{M}\right)&\ge& 
\frac{T_1\Nn_1^\ast|\Gamma|}{2M^2},
\end{eqnarray*}
where the inequality follows from the choice of $M$.
Since $M=O(\Nn_1+\Nn_2+\Nn_3)$ and $|\Gamma|\ge M^{1-\epsilon}$,
we conclude that $\E[\overline{|V_f}|]=\Omega\left(\frac{T_1 \Nn_1^\ast}{(\Nn_1+\Nn_2+\Nn_3)^{1+\epsilon}}\right)$.
\end{proof}
Theorem \ref{th_1} follows from Proposition \ref{prop_secondpruningv1}
by fixing a choice of the weights $\omega_0,\omega_1,\ldots,\omega_{N+1}$ for which
$|\overline{V_f}|\ge \E[|\overline{V_f}|]$ (such a choice necessarily exists from 
the definition of the expectation).

\subsection{The second pruning: second version and proof of Theorem \ref{th_2}}\label{sub_second2}
In this subsection $M$ is an arbitrary prime such that $64\Nn_2<M<128\Nn_2$.

The first version of the pruning described in Subsection \ref{sub_second1}
was designed to obtain an edgeless subgraph of $G$. In this subsection we 
describe how to modify it to obtain a union of many large 1-cliques instead.
The detailed procedure of the new pruning algorithm is described in Figure \ref{fig_alg2}. 
The only difference is that, at Step 2.2, a vertex is not added in $L$ only if it is 
connected to another vertex with the same second or third index.

\begin{figure}
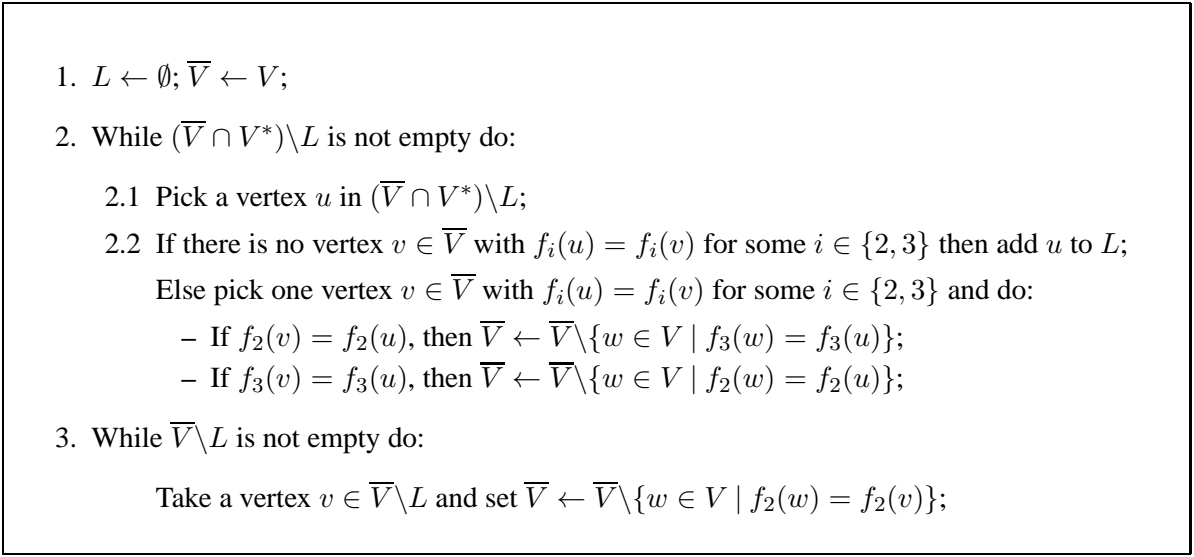

\begin{center}
    \begin{tabular}[h!]{|p{0.95\textwidth}|}
    \hline
\begin{enumerate}    
\item
$L\gets\emptyset$; $\overline{V}\gets V$;  
\item
While $(\overline{V}\cap V^\ast)\backslash L$ is not empty do: 
\begin{itemize}
\item[2.1]
Pick a vertex $u$ in $(\overline{V}\cap V^\ast)\backslash L$;
\item[2.2]
If there is no vertex $v\in \overline{V}$ with $f_i(u)=f_i(v)$ for some $i\in\{2,3\}$  then add $u$ to $L$;
\item[]
Else pick one vertex $v\in \overline{V}$ with $f_i(u)=f_i(v)$ for some $i\in\{2,3\}$ and do:
\begin{itemize}
\item
If $f_2(v)=f_2(u)$, then $\overline{V}\gets \overline{V}\backslash \{w\in V\:|\:f_3(w)=f_3(u)\}$;
\item
If $f_3(v)=f_3(u)$, then $\overline{V}\gets \overline{V}\backslash \{w\in V\:|\:f_2(w)=f_2(u)\}$;
\end{itemize}
\end{itemize}
\item
While $\overline{V}\backslash L$ is not empty do: 
\begin{itemize}
\item[]
Take a vertex $v\in \overline{V}\backslash L$ and set $\overline{V}\gets \overline{V}\backslash \{w\in V\:|\:f_2(w)=f_2(v)\}$;
\end{itemize}
\end{enumerate}
        \\\hline
    \end{tabular}\caption{The second pruning (second version)\label{fig_alg2}}
\end{center}
\end{figure}

Let $\overline{V_f}$ denote the contents of $\overline{V}$ at the end of the procedure.
By slightly modifying the arguments of the previous subsection, it is easy to see that
the resulting graph has only vertices from $V^\ast$ and is a disjoint union of 1-cliques
(i.e., each 
connected component is a 1-clique).
\begin{proposition}\label{prop_comp}
The subgraph of $G$ induced by $\overline{V_f}$ is a disjoint union of 1-cliques.
Moreover, $\overline{V_f}\subseteq V^\ast$ and 
$|\overline{V_f}|\ge |V^\ast|-|E''|$.
\end{proposition}
\begin{proof}
Let $L_f$ denote the contents of $L$ at the end of the procedure.
Due to Step 3, we know that $\overline{V_f}\subseteq L_f$.
Moreover, any vertex added to $L$ cannot be later removed from $\overline{V}$,
since it has no neighbor with the same second or third index.
Thus $L_f\subseteq\overline{V_f}$, and we conclude that  
$\overline{V_f}= L_f$, which shows in particular that $\overline{V_f}\subseteq V^\ast$.
Furthermore, since each vertex in $L$ has no neighbor with the same second or third index, 
the subgraph induced by $\overline{V_f}$ is a disjoint union of 1-cliques. 

The inequality $|\overline{V_f}|\ge |V^\ast|-|E''|$ is obtained 
as in the proof of Proposition \ref{prop_secondpruningv1},
but replacing the edge set $E'$ by~$E''$.
\end{proof}

We now give the proof of Theorem \ref{th_2}, by using Lemmas \ref{lemma_W} 
and \ref{lemma_R} to evaluate
the size and the numbers of 1-cliques in the resulting graph.
\begin{proof}[Proof of Theorem \ref{th_2}]
Using Lemma \ref{lemma_W} with the value $M>64\Nn_2$ and Markov's bound,
we conclude that 
$$\Pr\left[|E''|\le \frac{T_1\Nn_1^\ast |\Gamma|}{16M^2}\right]\ge \frac{1}{2}.$$

Note that the conditions $M<128\Nn_2$ and $\frac{\Nn_2T_1}{\Nn_1^\ast}+\frac{\Nn_2}{T_1}<\frac{1}{1024}$ imply
that 
$$
1-\frac{4MT_1}{\Nn_1^\ast}-\frac{4M}{T_1|\Gamma|}> 
1-512\left(\frac{\Nn_2T_1}{\Nn_1^\ast}+\frac{\Nn_2}{T_1|\Gamma|}\right)>
1-512\left(\frac{\Nn_2T_1}{\Nn_1^\ast}+\frac{\Nn_2}{T_1}\right)>
\frac{1}{2}.
$$
Thus the probability that a set $R$ as in Lemma \ref{lemma_R} exists and 
the inequality $|E''|\le \frac{T_1\Nn_1^\ast|\Gamma|}{16M^2}$ simultaneously holds
is positive. There then exists a choice of the weights $\omega_0,\omega_1,\ldots,\omega_{N+1}$
such that this happens. Let us take such a choice.


From Proposition \ref{prop_comp} we know that 
at most $|E''|$ vertices in $V^\ast$ are removed by the second pruning.
Then there 
exists a set 
$R'\subseteq R$ of size $|R'|\ge T_1|\Gamma|/(4M)$ such that, for any
$r'\in R'$, there are at least $\Nn_1^\ast/(4M)$ vertices $u$ with $f_1(u)=r'$ remaining
after the second pruning.
This is because, otherwise, 
from the properties of the set $R$ stated in Lemma \ref{lemma_R}
it would be necessary to remove more than
$T_1\Nn_1^\ast|\Gamma|/(16M^2)$ vertices during the second pruning.
\end{proof}

\section{Algorithm for Rectangular Matrix Multiplication}\label{sec_construction}
In this section we present our algorithm, which essentially consists in the two algorithmic 
steps described in Subsections~\ref{sub_step1} and~\ref{sub_step2}.
We first start by explaining the construction we will use.
\subsection{Our construction}\label{subsec_ourcon}
Let $a_{004},a_{400},a_{013},a_{103},a_{301},a_{022},a_{202}, a_{112},a_{211}$ be  
nine arbitrary positive\footnote{The hypothesis that each $a_{ijk}$ is not zero 
is made only for convenience  (all the bounds presented 
in this paper are obtained using positive values for these parameters). 
More specifically, this hypothesis is used only  
when approximating quantities like $(a_{ijk}N)!$ using Stirling's inequality.
Without the hypothesis it would be necessary to treat the (trivial) case $a_{ijk}= 0$ separately.} 
rational numbers 
such that
\begin{equation*}
2a_{004}+a_{400}+2a_{013}+2a_{103}+2a_{301}+a_{022}+2a_{202}+2a_{112}+a_{211}=1
\end{equation*}
and 
\begin{equation*}
a_{013}a_{202}a_{112}=a_{103}a_{022}a_{211}.
\end{equation*}

Let us define rational numbers $A_0,A_1,A_2,A_3,A_4,B_0,B_1,B_2,B_3,B_4$ as follows:
\begin{eqnarray*}
A_0&=&
2a_{004}+2a_{013}+a_{022}\\
A_1&=&
2a_{103}+2a_{112}\\
A_2&=&
2a_{202}+a_{211}\\
A_3&=&
2a_{301}\\
A_4&=&
a_{400}\\
B_0&=&
a_{004}+a_{400}+a_{103}+a_{301}+a_{202}\\
B_1&=&
a_{013}+a_{301}+a_{112}+a_{211}\\
B_2&=&
a_{022}+a_{202}+a_{112}\\
B_3&=&
a_{013}+a_{103}\\
B_4&=&
a_{004}.
\end{eqnarray*}
Note that $\sum_{i=0}^4A_i=\sum_{i=0}^4B_i=1.$

It will be convenient to define six additional numbers $a_{040}$, $a_{031}$, $a_{130}$, 
$a_{310}$, $a_{220}$ and  $a_{121}$ as
$a_{040}=a_{004}$, $a_{031}=a_{013}$,
$a_{130}=a_{103}$, $a_{310}=a_{301}$, $a_{220}=a_{202}$ and $a_{121}=a_{112}$. 
We can then rewrite concisely the $A_{i}$'s and the $B_j$'s as follows.
\begin{eqnarray*}
A_i&=&\sum_{\begin{subarray}{c}0\le j,k\le 4\\ i+j+k=4\end{subarray}} a_{ijk}\hspace{3mm}\textrm{ for }  i=0,1,2,3,4\\
B_j&=&\sum_{\begin{subarray}{c}0\le i,k\le 4\\ i+j+k=4\end{subarray}} a_{ijk}\hspace{3mm}\textrm{ for }  j=0,1,2,3,4\\
\end{eqnarray*}

Let $N$ be a large enough positive integer such each $Na_{ijk}$ is an integer.
We rise the construction $F_q\otimes F_q$ described in Section \ref{section_CW} to the $N$-th power.
Observe that $(F'_q\otimes F'_q)^{\otimes N} \degen (F_q\otimes F_q)^{\otimes N}$ and
\begin{eqnarray*}
(F'_q\otimes F'_q)^{\otimes N}&=&
\sum_{\begin{subarray}{l} IJK\end{subarray}}T_{IJK},
\end{eqnarray*}
where the sum is over all triples of sequences $IJK$ with $I,J,K\in \{0,1,2,3,4\}^N$ such that $I_\ell+J_\ell+K_\ell=4$ for all $\ell\in\{1,\ldots,N\}$. 
Here we use the notation
$T_{IJK}=T_{I_1J_1K_1}\otimes \cdots\otimes T_{I_NJ_NK_N}$. Note that there are $15^N$ terms $T_{IJK}$ in the above sum.
In the tensor product the number of $x$-variables is $(q+2)^{2N}$. The number of $y$-variables and $z$-variables 
is also $(q+2)^{2N}$.
Remember that in the original construction, each $x$-variable was indexed 
by a superscript in $\{0,1,2,3,4\}$. Each $x$-variable in the tensor product is thus indexed by a sequence of $N$ 
such superscripts, i.e., by an element $I\in \{0,1,2,3,4\}^N$.
The same is true for 
the $y$-variables and the $z$-variables. Note that the $x$-variables appearing in $T_{IJK}$ have superscript~$I$, the $y$-variables appearing in $T_{IJK}$ have superscript $J$, and the $z$-variables appearing in 
$T_{IJK}$ have superscript~$K$.

Let us introduce the following definition.
\begin{definition}\label{def_type}
Let 
$\overline{a}_{004}$, $ \overline{a}_{040}$, $\overline{a}_{400}$, $\overline{a}_{013}$, $\overline{a}_{031}$, $\overline{a}_{103}$, $\overline{a}_{130}$, $\overline{a}_{301}$, $\overline{a}_{310}$, $\overline{a}_{022}$, $\overline{a}_{202}$, $\overline{a}_{220}$, $\overline{a}_{112}$, $\overline{a}_{121}$, $\overline{a}_{211}$ 
be fifteen nonnegative rational numbers.
We say that a triple $IJK$ is of type $[\overline{a}_{ijk}]$ if 
$$|\left\{\ell\in\{1,\ldots,N\}\:|\: I_\ell=i, J_\ell=j \textrm{ and } K_\ell=k\right\}|=\overline{a}_{ijk}N$$
for all 15 combinations of positive $i,j,k$ with $i+j+k=4$.
\end{definition}
With a slight abuse of notation, we will say that a form $T_{IJK}$ is of type $[\overline{a}_{ijk}]$ if the triple $IJK$ is of type $[\overline{a}_{ijk}]$.


\subsection{The first step}\label{sub_step1}

We set to zero all $x$-variables except those satisfying the following condition: their superscript $I$ 
has exactly $A_0N$ coordinates with value $0$, $A_1N$ coordinates with value $1$, $A_2N$ coordinates with value $2$, $A_3N$ coordinates with value $3$ and $A_4N$ coordinates with value $4$.
We will say that such a sequence $I$ is of type $A$.
There are 
$$
T_X={N\choose{A_0N,\ldots,A_4N}}=\Theta\left(\frac{1}{N^{2}}\left(\frac{1}{A_0^{A_0}A_1^{A_1}A_2^{A_2}A_3^{A_3}A_4^{A_4}}\right)^{N}\right)
$$
sequences $I$ of type $A$ (the approximation is done using Stirling's formula). 
After the zeroing operation, all forms $T_{IJK}$ such that $I$ is not of type $A$ disappear (i.e., become zero).

We process the $y$-variables and the $z$-variables slightly differently.
We set to zero all $y$-variables except those satisfying the following condition: their superscript $J$ 
has exactly $B_0N$ coordinates with value $0$, $B_1N$ coordinates with value $1$, $B_2N$ coordinates with value $2$,  $B_3N$ coordinates with value $3$ and $B_4N$ coordinates with value $4$.
We will say that such a sequence is of type $B$.
There are 
$$
T_Y={N\choose{B_0N,\ldots,B_4N}}=\Theta\left(\frac{1}{N^{2}}\left(\frac{1}{B_0^{B_0}B_1^{B_1}B_2^{B_2}B_3^{B_3}B_4^{B_4}}\right)^{N}\right)
$$
sequences $J$ of type $B$. 
Similarly, we set to zero all $z$-variables except those such that their superscript $K$ is of type $B$ (there are $T_Y$ such sequences). 

After these three zeroing operations, 
the forms $T_{IJK}$ remaining are precisely those such that $I$ is of type $A$, $J$ is of type $B$, and $K$ is of type $B$. 
Equivalently, the forms remaining are precisely the forms $T_{IJK}$ that are of type $[\overline{a}_{ijk}]$ with 
fifteen numbers 
$\overline{a}_{ijk}$
(for all fifteen combinations of positive $i,j,k$ such that $i+j+k=4$)
satisfying the following four conditions:
\begin{align}
\label{con1}
\hspace{5mm}&\overline{a}_{ijk}N\in\{0,1,\ldots,N\} \:\:\:\textrm{for all }i,j,k;\\
\label{con2}
\hspace{5mm}&A_i=\sum_{j,k \::\: i+j+k=4} \overline{a}_{ijk}\hspace{3mm}\textrm{ for }  i=0,1,2,3,4;\\
\label{con3}
\hspace{5mm}&B_j=\sum_{i,k \::\: i+j+k=4} \overline{a}_{ijk}\hspace{3mm}\textrm{ for }  j=0,1,2,3,4;\\
\label{con4}
\hspace{5mm}&B_k=\sum_{i,j \::\: i+j+k=4} \overline{a}_{ijk}\hspace{3mm}\textrm{ for }  k=0,1,2,3,4.
\end{align}

Let $I$ be a fixed sequence of type $A$. The number of non-zeros forms $T_{IJK}$ with this sequence $I$ as its first index is
thus precisely
\begin{eqnarray*}
\Nn_X&=&\sum_{[\overline{a}_{ijk}]}
{A_0N \choose \overline{a}_{004}N,\overline{a}_{040}N,\overline{a}_{013}N,\overline{a}_{031}N,\overline{a}_{022}N}
{A_1N \choose \overline{a}_{103}N,\overline{a}_{130}N,\overline{a}_{112}N,\overline{a}_{121}N}\times\\
&&\hspace{48mm}{A_2N \choose \overline{a}_{202}N,\overline{a}_{220}N,\overline{a}_{211}N}
{A_3N \choose \overline{a}_{301}N,\overline{a}_{310}N}
{A_4N \choose \overline{a}_{400}N},
\end{eqnarray*}
where the sum is over all the choices of fifteen parameters $\overline{a}_{ijk}$'s satisfying conditions (\ref{con1})--(\ref{con4}).

For a fixed sequence $J$ of type $B$, the number of non-zeros forms $T_{IJK}$
with this sequence $J$ as its second index is
\begin{eqnarray*}
\Nn_Y&=&\sum_{[\overline{a}_{ijk}]}
{B_0N \choose \overline{a}_{004}N,\overline{a}_{400}N,\overline{a}_{103}N,\overline{a}_{301}N,\overline{a}_{202}N}
{B_1N \choose \overline{a}_{013},\overline{a}_{310},\overline{a}_{112},\overline{a}_{211}N}\times\\
&&\hspace{38mm}{B_2N \choose \overline{a}_{022}N,\overline{a}_{220}N,\overline{a}_{121}N}
{B_3N \choose \overline{a}_{031}N,\overline{a}_{130}N}
{B_4N \choose \overline{a}_{040}N},
\end{eqnarray*}
where the sum is again over all the choices of fifteen parameters $\overline{a}_{ijk}$'s satisfying conditions (\ref{con2})--(\ref{con4}).
Similarly, for a fixed sequence $K$ of type $B$, the number of non-zeros forms $T_{IJK}$
with this sequence $K$ as its third index is
\begin{eqnarray*}
\Nn_Z&=&\sum_{[\overline{a}_{ijk}]}
{B_0N \choose \overline{a}_{040}N,\overline{a}_{400}N,\overline{a}_{130}N,\overline{a}_{310}N,\overline{a}_{220}N}
{B_1N \choose \overline{a}_{031},\overline{a}_{301},\overline{a}_{121},\overline{a}_{211}N}\times\\
&&\hspace{38mm}{B_2N \choose \overline{a}_{022}N,\overline{a}_{202}N,\overline{a}_{112}N}
{B_3N \choose \overline{a}_{013}N,\overline{a}_{103}N}
{B_4N \choose \overline{a}_{004}N}.
\end{eqnarray*}

The total number of remaining triples is 
\begin{equation}\label{relation1}
T_X\Nn_X=T_Y\Nn_Y=T_Y\Nn_Z.
\end{equation}
Note that this implies that $\Nn_Y=\Nn_Z$.

We will also be interested in the number of remaining forms $T_{IJK}$ 
of type $[a_{ijk}]$. 
For a fixed sequence~$I$ of type $A$, 
the number of non-zeros forms $T_{IJK}$ of type $[a_{ijk}]$ with this sequence $I$ as its first index is
\begin{eqnarray*}
\Nn_X^\ast&=&
{A_0N \choose a_{004}N,a_{004}N,a_{013}N,a_{013}N,a_{022}N}
{A_1N \choose a_{103}N,a_{103}N,a_{112}N,a_{112}N}\times\\
&&\hspace{48mm}{A_2N \choose a_{202}N,a_{202}N,a_{211}N}
{A_3N \choose a_{301}N,a_{301}N}
{A_4N \choose a_{400}N}.
\end{eqnarray*}
For a fixed sequence $J$ of type $B$, 
the number of non-zeros forms $T_{IJK}$ of type $[a_{ijk}]$ with this sequence $J$ as its second index is
\begin{eqnarray*}
\Nn_Y^\ast&=&
{B_0N \choose a_{004}N,a_{400}N,a_{103}N,a_{301}N,a_{202}N}
{B_1N \choose a_{013},a_{301},a_{112},a_{211}N}\times\\
&&\hspace{38mm}{B_2N \choose a_{022}N,a_{202}N,a_{112}N}
{B_3N \choose a_{013}N,a_{103}N}
{B_4N \choose a_{004}N}.
\end{eqnarray*}
We have 
\begin{equation}\label{relation2}
T_X\Nn^\ast_X=T_Y\Nn_Y^\ast.
\end{equation}

\subsection{Approximation}
In this subsection we will use the notation $[\overline{a}_{ijk}]$ to represent an arbitrary set of fifteen parameters~$\overline a_{ijk}$
such that $0\le \overline a_{ijk}\le 1$ for each $i,j,k$.
Let $c([\overline a_{ijk}])$ denote the number of nonzero elements among these fifteen parameters.
Consider the following expression:
$$g([\overline{a}_{ijk}])=
\left(\overline{a}_{004}^{\overline{a}_{004}}\overline{a}_{040}^{\overline{a}_{040}}\overline{a}_{400}^{\overline{a}_{400}}
\overline{a}_{013}^{\overline{a}_{013}}\overline{a}_{031}^{\overline{a}_{031}}\overline{a}_{103}^{\overline{a}_{103}}\overline{a}_{130}^{\overline{a}_{130}}\overline{a}_{301}^{\overline{a}_{301}}\overline{a}_{310}^{\overline{a}_{310}}
\overline{a}_{022}^{\overline{a}_{022}}\overline{a}_{202}^{\overline{a}_{202}}\overline{a}_{220}^{\overline{a}_{220}}
\overline{a}_{112}^{\overline{a}_{112}}\overline{a}_{121}^{\overline{a}_{121}}\overline{a}_{211}^{\overline{a}_{211}}\right)^{-1}.
$$
Using Stirling's formula, we can give the following approximations.
\begin{eqnarray*}
\Nn_X
&=&\Theta\left(\sum_{[\overline{a}_{ijk}]}\frac{\left[ A_0^{A_0}A_1^{A_1}A_2^{A_2}A_3^{A_3}A_4^{A_4}\times g([\overline{a}_{ijk}])\right]^{N}}{N^{(c([\overline a_{ijk}])-5)/2}}\right)\\
\Nn_Y
&=&\Theta\left(\sum_{[\overline{a}_{ijk}]}\frac{\left[ B_0^{B_0}B_1^{B_1}B_2^{B_2}B_3^{B_3}B_4^{B_4}\times g([\overline{a}_{ijk}])\right]^{N}}{N^{(c([\overline a_{ijk}])-5)/2}}\right)\\
\Nn^\ast_X&=&
\Theta\left(\frac{1}{N^5}\left[ A_0^{A_0}A_1^{A_1}A_2^{A_2}A_3^{A_3}A_4^{A_4}\times g([a_{ijk}])\right]^{N}\right)\\
\Nn^\ast_Y&=&
\Theta\left(\frac{1}{N^5}\left[ B_0^{B_0}B_1^{B_1}B_2^{B_2}B_3^{B_3}B_4^{B_4}\times g([a_{ijk}])\right]^{N}\right)
\end{eqnarray*}
The first two sums are again 
over all the choices of fifteen parameters $\overline{a}_{ijk}$'s satisfying conditions (\ref{con1})--(\ref{con4}).
Note that, for any $[\overline a_{ijk}]$ satisfying these four conditions, $c([\overline a_{ijk}])\ge 5$, since the 
$A_i$'s are non-zero.

We know that $\Nn_X^\ast\le \Nn_X$ and $\Nn_Y^\ast\le \Nn_Y$, by definition. The following proposition shows that $\Nn_X$ and
$\Nn_Y$ can actually be approximated by $\Nn_X^\ast$ and $\Nn_Y^\ast$, respectively.
\begin{proposition}\label{proposition_max}
$\Nn_X=O(N^8\Nn_X^\ast)
\:\:\textrm{ and }\:\: 
\Nn_Y=O(N^8\Nn_Y^\ast)$.
\end{proposition}
\begin{proof}
Any set of values $[\overline{a}_{ijk}]$ that satisfies conditions (\ref{con1})--(\ref{con4}) is such that
$\overline{a}_{004}=a_{004}$, $\overline{a}_{040}=a_{040}$ and $\overline{a}_{400}=a_{400}$.
Moreover, the other values $a_{ijk}$ depend only on $\overline{a}_{103},\overline{a}_{031}$ and $\overline{a}_{301}$:
\begin{eqnarray*}
\overline{a}_{013}&=&B_3N-\overline{a}_{103}\\
\overline{a}_{130}&=&B_3N-\overline{a}_{031}\\
\overline{a}_{310}&=&A_3N-\overline{a}_{301}\\
\overline{a}_{022}&=&(A_0-2B_4-B_3)N+\overline{a}_{103}-\overline{a}_{031}\\
\overline{a}_{202}&=&(-A_4+B_0-B_4)N-\overline{a}_{103}-\overline{a}_{301}\\
\overline{a}_{220}&=&(-A_3-A_4+B_0-B_3-B_4)N+\overline{a}_{031}+\overline{a}_{301}\\
\overline{a}_{112}&=&(-A_0+A_4-B_0+B_2+B_3+3B_4)N+\overline{a}_{031}+\overline{a}_{301}\\
\overline{a}_{121}&=&(-A_0+A_3+A_4-B_0+B_2+2B_3+3B_4)N-\overline{a}_{103}-\overline{a}_{301}\\
\overline{a}_{211}&=&(A_2+A_3+2A_4-2B_0+B_3+2B_4)N-\overline{a}_{031}+\overline{a}_{103}.
\end{eqnarray*}
Note that there are at most $(N+1)$ choices for each $\overline{a}_{103},\overline{a}_{031}$ and $\overline{a}_{301}$,
from condition (\ref{con1}).
There are thus at most $(N+1)^3$ choices for the values $[\overline{a}_{ijk}]$ satisfying conditions (\ref{con1})--(\ref{con4}).

We now show that the expression $g([\overline{a}_{ijk}])$
is maximized for the values $[\overline{a}_{ijk}]=[a_{ijk}]$.
Let us take the logarithm of the expression $g$. 
Since this is a concave function on a convex domain,
a local optimum of $\log g$ is a global maximum of $g$.
The partial derivatives of $\log g$ are as follows.
\begin{eqnarray*}
\frac{\partial \log g}{\partial\overline{a}_{103}}&=&\phantom{-}\log(\overline{a}_{013})-\log(\overline{a}_{103})-\log(\overline{a}_{022})+\log(\overline{a}_{202})+\log(\overline{a}_{121})-\log(\overline{a}_{211})\\
\frac{\partial\log g}{\partial\overline{a}_{031}}&=&-\log(\overline{a}_{031})+\log(\overline{a}_{130})+\log(\overline{a}_{022})-\log(\overline{a}_{220})-\log(\overline{a}_{112})+\log(\overline{a}_{211})\\
\frac{\partial\log g}{\partial\overline{a}_{301}}&=&-\log(\overline{a}_{301})+\log(\overline{a}_{310})+\log(\overline{a}_{202})-\log(\overline{a}_{220})-\log(\overline{a}_{112})+\log(\overline{a}_{121})\\
\end{eqnarray*}
The values $[\overline{a}_{ijk}]=[a_{ijk}]$
satisfy $\frac{\partial\log g}{\partial\overline{a}_{103}}=\frac{\partial\log g}{\partial\overline{a}_{031}}=\frac{\partial\log g}{\partial\overline{a}_{301}}=0$
since $a_{013}a_{202}a_{112}=a_{103}a_{022}a_{211}$.

We conclude that 
$$\Nn_X=O\left((N+1)^3\left[ A_0^{A_0}A_1^{A_1}A_2^{A_2}A_3^{A_3}A_4^{A_4}\times g([a_{ijk}])\right]^{N}\right)=O\left(N^8 \Nn_X^\ast\right),$$
and similarly $\Nn_Y=O\left(N^8 \Nn_Y^\ast\right)$.
\end{proof}

\subsection{The second step}\label{sub_step2}
We will now apply the results of Section \ref{sec_graphtheory}, by associating to $U$ the set of 
all forms $T_{IJK}$ that are 
of type $[\overline a_{ijk}]$ for values $\overline a_{ijk}$ 
satisfying conditions (\ref{con1})--(\ref{con4}),
and to $U^\ast$ the set of 
all forms $T_{IJK}$ of type $[a_{ijk}]$. Note that all the conditions of 
Subsection \ref{sub_graphtheorysetting} are satisfied:
we have $\tau=4$, and values $T_1=T_X$, $T_2=T_3=T_Y$,
$\Nn_1=\Nn_X$, $\Nn_2=\Nn_3=\Nn_Y$, $\Nn_1^\ast=\Nn_X^\ast$ and
$\Nn_2^\ast=\Nn_3^\ast=\Nn_Y^\ast$.

With this association we can use the graph-theoretic framework developed in 
Section \ref{sec_graphtheory}. 
The initial graph $G$,
as defined in Subsection \ref{sub_graphtheorysetting}, corresponds 
to the current sum of trilinear forms (each vertex corresponds to a form 
$T_{IJK}$ of type $[\overline a_{ijk}]$ for values $\overline a_{ijk}$ 
satisfying conditions (\ref{con1})--(\ref{con4})). A removal operation on $G$
corresponds
to zeroing variables with a given superscript.
For instance, removing all vertices $u\in G$ such that 
$f_1(u)=I$ corresponds to zeroing all the $x$-variables 
with superscript $I$. 
Our goal is to zero variables in order 
to obtain a sum of several forms $T_{IJK}$ satisfying the 
following two conditions. First,
each form in the sum is of type $[a_{ijk}]$.
Then, the forms in the 
sum do not share any index 
(i.e., if $T_{IJK}$ and $T_{I'J'K'}$ are in the sum, 
then $I\neq I'$, $J\neq J'$ and $K\neq K'$), which implies
that the same variable does not appear in more than one 
form, and thus that the sum is direct.
This corresponds to constructing an edgeless 
subgraph of $G$ in which all the vertices are in $U^\ast$,
so we are in a situation where Theorem \ref{th_1}
can be applied.

Suppose that the inequality 
$$A_0^{A_0}A_1^{A_1}A_2^{A_2}A_3^{A_3}A_4^{A_4}\ge B_0^{B_0}B_1^{B_1}B_2^{B_2}B_3^{B_3}B_4^{B_4}$$ holds.
Then $T_X=O(T_Y)$, and Equalities (\ref{relation1}) and  (\ref{relation2}) imply that $\Nn_Y=O(\Nn_X)$ and $\Nn_Y^\ast=O(\Nn_X^\ast)$.
From Proposition \ref{proposition_max}
we then obtain the relation $\Nn_1+\Nn_2+\Nn_3=O(N^8\Nn_X^\ast)$.
By the above discussion and Theorem \ref{th_1}, we can obtain
a direct sum of 
$$\Omega\left(\frac{T_X\Nn^\ast_X}{(N^8\Nn_X^\ast)^{1+\epsilon}}\right)$$ 
forms, all of type $[a_{ijk}]$.
By using the trivial 
upper bound $\Nn_X^\ast\le  15^N$,
we obtain the following theorem.
\begin{theorem}\label{theorem_consresult}
Let $q$ be any positive integer and 
$a_{004}$, $a_{400}$, $a_{013}$, $a_{103}$, $a_{301}$, $a_{022}$, $a_{202}$,  $a_{112}$ and $a_{211}$ be any nine positive rational numbers 
satisfying the following three conditions:
\begin{itemize}
\item
$2a_{004}+a_{400}+2a_{013}+2a_{103}+2a_{301}+a_{022}+2a_{202}+2a_{112}+a_{211}=1$;
\item
$a_{013}a_{202}a_{112}=a_{103}a_{022}a_{211}$;
\item
$A_0^{A_0}A_1^{A_1}A_2^{A_2}A_3^{A_3}A_4^{A_4}\ge B_0^{B_0}B_1^{B_1}B_2^{B_2}B_3^{B_3}B_4^{B_4}$.
\end{itemize}
Then,
for any constant $\epsilon>0$, the trilinear form $(F_q\otimes F_q)^{\otimes N}$ can be converted (i.e., degenerated in the sense of Definition \ref{def_deg}) into a direct sum of 
$$ 
\Omega\left(\frac{1}{N^{10 +8\epsilon}15^{N\epsilon}}\left[\frac{1}{A_0^{A_0}A_1^{A_1}A_2^{A_2}A_3^{A_3}A_4^{A_4}}\right]^{N}\right)
$$
forms, each form being isomorphic to 
$$
\bigotimes_{\begin{subarray}{c}0\le i,j,k\le 4\\ i+j+k=4\end{subarray}}T_{ijk}^{\otimes a_{ijk}N}.
$$
\end{theorem}

\section{Upper Bounds on the Exponent of Rectangular Matrix Multiplication}\label{sec_analysis}
Theorem \ref{theorem_consresult} showed how the form 
$(F_q\otimes F_q)^{\otimes N}$ can be used to obtain a direct sum of
many forms $T_{IJK}$ such that
$$
T_{IJK}\cong \bigotimes_{i,j,k:i+j+k=4}T_{ijk}^{\otimes a_{ijk}N}.
$$
In order to apply Sch\"onhage's asymptotic sum inequality (Theorem \ref{theorem_schonhage}), 
we need to analyze the smaller forms $T_{ijk}$.
All the forms except $T_{112}$, $T_{121}$ and $T_{211}$ correspond to matrix multiplications, as described
in Section~\ref{section_CW}. 
In Subsection \ref{sub_T112} we analyze the forms $T_{112}$, $T_{121}$ and $T_{211}$.
Then, in Subsection \ref{sub_together},
we put all our results together and prove our main result.

\subsection{The forms {\boldmath$T_{112}$}, {\boldmath$T_{121}$} and {\boldmath$T_{211}$}}\label{sub_T112}
Let us recall the definition of these three forms.
\begin{eqnarray*}
T_{112}&=&\sum_{i=1}^q x_{i,0}^{1}y_{i,0}^{1}z_{0,q+1}^{2}+\sum_{k=1}^q x_{0,k}^{1}y_{0,k}^{1}z_{q+1,0}^{2}+\sum_{i,k=1}^q x_{i,0}^{1}y_{0,k}^{1}z_{i,k}^{2}+\sum_{i,k=1}^q x_{0,k}^{1}y_{i,0}^{1}z_{i,k}^{2}\\
T_{121}&=&\sum_{i=1}^q x_{i,0}^{1}y_{0,q+1}^{2}z_{i,0}^{1}+
\sum_{k=1}^q x_{0,k}^{1}y_{q+1,0}^{2}z_{0,k}^{1}+
\sum_{i,k=1}^q x_{i,0}^{1}y_{i,k}^{2}z_{0,k}^{1}+
\sum_{i,k=1}^q x_{0,k}^{1}y_{i,k}^{2}z_{i,0}^{1}\\
T_{211}&=&\sum_{i=1}^q x_{0,q+1}^{2}y_{i,0}^{1}z_{i,0}^{1}+
\sum_{k=1}^q x_{q+1,0}^{2}y_{0,k}^{1}z_{0,k}^{1}+
\sum_{i,k=1}^q x_{i,k}^{2}y_{i,0}^{1}z_{0,k}^{1}+
\sum_{i,k=1}^q x_{i,k}^{2}y_{0,k}^{1}z_{i,0}^{1}
\end{eqnarray*}

We first focus on the form $T_{211}$. 
It will be convenient to write $T_{211}=t_{011}+t_{101}+t_{110}+t_{200}$, where
\begin{eqnarray*}
t_{011}&=&\sum_{i=1}^q x_{0,q+1}^{0}y_{i,0}^{1}z_{i,0}^{1},\\
t_{101}&=&\sum_{i,k=1}^q x_{i,k}^{1}y_{0,k}^{0}z_{i,0}^{1},\\
t_{110}&=&\sum_{i,k=1}^q x_{i,k}^{1}y_{i,0}^{1}z_{0,k}^{0},\\
t_{200}&=&\sum_{k=1}^q x_{q+1,0}^{2}y_{0,k}^{0}z_{0,k}^{0}.
\end{eqnarray*}
Note that the superscripts in these forms differ from the original superscripts.
They are nevertheless uniquely determined by the subscripts of the variables.
Observe that $t_{011}\cong t_{200}\cong\braket{1,1,q}$, 
which corresponds to the product of a scalar by a row vector,
and $t_{101}\cong t_{110}\cong\braket{q,q,1}$, which corresponds to 
the product of a $q\times q$ matrix by a column vector.

The following proposition states that tensor powers of $T_{211}$ can be used to construct
a direct sum of several trilinear forms, each one being a $\Cc$-tensor in which the support and all the 
components are isomorphic
to a rectangular matrix product.

\begin{proposition}\label{prop_UH}
Let $b$ be any constant such that $0.916027<b\le1$. 
Then there exists a constant $c\ge 1$ depending only on $b$ such that, 
for any $\epsilon>0$ and any large enough integer $m$, the form $T_{211}^{\otimes 2m}$ can be 
converted into a direct sum of 
$$
\Omega\left(\frac{1}{mc^{2\epsilon m}}\cdot \left[\frac{2}{(2b)^{b}(1-b)^{1-b}}\right]^{2m}\right)
$$
trilinear forms, each form being a $\Cc$-tensor in which:
\begin{itemize}
\item
each component is isomorphic to $\braket{q^{2bm},q^{2bm},q^{2(1-b)m}}$;
\item
the support is isomorphic to $\suppc(\braket{1,1,H})$, where
$
H=\Omega\left(\frac{1}{\sqrt{m}}\cdot\left[(2b)^b(1-b)^{(1-b)}\right]^{2m}\right).
$
\end{itemize}
\end{proposition}

\noindent\textbf{Remark.} Proposition \ref{prop_UH} uses the convention $0^0=1$. 
For the case $b=1$, the proposition states that the form $T_{112}^{\otimes 2m}$ 
can be used to construct at least one tensor with support isomorphic to
$\suppc(\braket{1,1,H})$ for 
$H=\Omega\left(4^m/\sqrt{m}\right)$, each component being isomorphic to 
$\braket{q^{2m},q^{2m},1}$. 

\begin{proof}[Proof of Proposition \ref{prop_UH}]
For simplicity we suppose that $bm$ is an integer 
(otherwise, we can work with $\floor{bm}$, which gives the same asymptotic complexity).

Let $S$ be the set of all triples $IJK$ with $I\in\{0,1,2\}^{2m}$ and $J,K\in\{0,1\}^{2m}$ such that
$I_\ell+J_\ell+K_\ell=2$ for all $\ell\in\{1,\ldots,2m\}$.
We rise the form $T_{211}$ to the $2m$-th power.
This gives the form
$$
\sum_{IJK\in S}t_{IJK}
$$
where $t_{IJK}=t_{I_1J_1K_1}\otimes \cdots\otimes t_{I_{2m}J_{2m}K_{2m}}$.
Each $x$-variable in the tensor product has for superscript a sequence in $\{0,1,2\}^{2m}$, 
and each $y$-variable or $z$-variable has for superscript a sequence in $\{0,1\}^{2m}$.
Let us decompose the space of $x$-variables as $\bigoplus_{I\in\{0,1,2\}^{2m}} X_I$, where $X_I$ denotes the subspace of $x$-variables with superscript $I$.
Similarly, decompose the space of $y$-variables as $\bigoplus_{J\in\{0,1\}^{2m}} Y_J$ and the space of $z$-variables as $\bigoplus_{K\in\{0,1\}^{2m}} Z_K$.  
The form
$\sum t_{IJK}$ above is then a $\Cc$-tensor with respect to this decomposition, with support
$S$.
We will now modify this form (by zeroing variables, which will modify its support) in order to obtain a simple expression for each of its components.

We zero all the $x$-variables except those for which the superscript $I$
has $(1-b)m$ coordinates with value $0$,
$(1-b)m$ coordinates with value  $2$ and $2bm$ coordinates with value $1$.
We zero all the $y$-variables and $z$-variables except those for which the superscript 
has $m$ coordinates with value $0$ and $m$ coordinates with value  $1$.
After these zeroing operations, the only forms remaining in the sum
are those corresponding to indexes $IJK$ satisfying the following four conditions.
\begin{eqnarray*}
|\left\{\ell\in\{1,\ldots,2m\}\:|\: I_\ell=0, J_\ell=1 \textrm{ and } K_\ell=1\right\}|&=&(1-b)m\\
|\left\{\ell\in\{1,\ldots,2m\}\:|\: I_\ell=1, J_\ell=0 \textrm{ and } K_\ell=1\right\}|&=&bm\\
|\left\{\ell\in\{1,\ldots,2m\}\:|\: I_\ell=1, J_\ell=1 \textrm{ and } K_\ell=0\right\}|&=&bm\\
|\left\{\ell\in\{1,\ldots,2m\}\:|\: I_\ell=2, J_\ell=0 \textrm{ and } K_\ell=0\right\}|&=&(1-b)m
\end{eqnarray*}
This means that each form $t_{IJK}$ in this new sum (i.e., each component in the corresponding $\Cc$-tensor) is isomorphic to 
$$t_{011}^{\otimes (1-b)m}\otimes t_{101}^{\otimes bm}\otimes t_{110}^{\otimes bm}\otimes t_{200}^{\otimes (1-b)m}
\cong \braket{q^{2bm},q^{2bm},q^{2(1-b)m}}.$$ 
We now analyze the support of the new sum (the decomposition considered is unchanged).

%

The case $b=1$ can be analyzed directly: there are ${2m \choose m}=\Theta\left(4^m/ \sqrt{m}\right)$ forms in the sum,
all of them sharing the same first index (the all-one sequence $\boldsymbol 1=1\cdots 1$). 
This sum is then $\sum_{J}t_{\boldsymbol 1JK}$, where for each $t_{\boldsymbol 1JK}$ the sequence $K$
is uniquely determined by $J$.
The support of the sum is thus isomorphic
to $\suppc(\braket{1,1,{2m \choose m}})$.

To analyze the case $b<1$, we will interpret the sum in the framework developed in Section \ref{sec_graphtheory}, 
by letting $U$ be the set of triples $IJK$ satisfying the above four conditions.
Indeed, all the requirements for $U$ are satisfied: we have $\tau=2$ and
\begin{eqnarray*}
T_1&=&{2m \choose (1-b)m,(1-b)m,2mb}=\Theta\left(\frac{1}{m}\cdot \left[\frac{2}{(2b)^b(1-b)^{1-b}}\right]^{2m}\right)\\
\Nn_1&=&{2mb \choose mb}=\Theta\left(\frac{1}{\sqrt{m}}\cdot  \left[2^b\right]^{2m}\right)\\
T_2=T_3&=&{2m \choose m}=\Theta\left(\frac{1}{\sqrt{m}}\cdot  \left[2\right]^{2m}\right)\\
\Nn_2=\Nn_3&=&{m \choose m(1-b)}{m \choose m(1-b)}=\Theta\left(\frac{1}{m}\cdot  \left[\frac{1}{b^b(1-b)^{1-b}}\right]^{2m}\right) .\\
\end{eqnarray*}
Note in particular that $N_1> N_2$, since $(2b)^b(1-b)^{1-b}>1$ for any $b\ge 0.773$.
Choose $U^\ast=U$ (which means that $\Nn_i=\Nn_i^\ast$ for each $i\in\{1,2,3\}$). 
The correspondence with the graph-theoretic interpretation of Section \ref{sec_graphtheory}
is as follows. Each vertex in the graph $G$ defined in Subsection \ref{sub_graphtheorysetting}
corresponds to one form $T_{IJK}$ in the sum, which is isomorphic 
to $\braket{q^{2bm},q^{2bm},q^{2(1-b)m}}$ from the discussion above. One removal
operation corresponds to zeroing variables.
A 1-clique of length $n$ corresponds to a sum of $n$ forms sharing their first index, which is precisely
a $\Cc$-tensor with support isomorphic to $\suppc(\braket{1,1,n})$. 


For any value $b>0.916027$ we have
$$
\frac{2}{(2b)^{2b}(1-b)^{2(1-b)}}<1
$$
and thus $T_1\Nn_2/\Nn_1=o(1)$. 
By Theorem \ref{th_2}, for any $\epsilon>0$ we can then convert the sum into a direct sum of  
$$\Omega\left(\frac{T_1}{\Nn_2^\epsilon}\right)=\Omega\left(\frac{1}{m^{1-\epsilon}}\cdot \left[\frac{2}{2^b\left(b^{b}(1-b)^{1-b}\right)^{1-\epsilon}}\right]^{2m}\right)$$ 
$\Cc$-tensors, each tensor having support isomorphic to 
$\suppc(\braket{1,1,n})$ for 
$$n=\Omega\left(\frac{\Nn_1}{\Nn_2}\right)=
\Omega\left(\sqrt{m}\cdot\left[(2b)^b(1-b)^{(1-b)}\right]^{2m}\right)$$
and components isomorphic to $\braket{q^{2bm},q^{2bm},q^{2(1-b)m}}$.
Finally, note that 
$$
\frac{1}{m^{1-\epsilon}}\cdot \left[\frac{2}{2^b\left(b^{b}(1-b)^{1-b}\right)^{1-\epsilon}}\right]^{2m}\ge
\frac{1}{mc^{2\epsilon m}}\left[\frac{2}{(2b)^{b}(1-b)^{1-b}}\right]^{2m}
$$
for some constant $c\ge 1$ depending only on $b$, since $b^{b}(1-b)^{1-b}<1$.
\end{proof}


The forms $T_{112}$ and $T_{121}$ can be analyzed in the same way as
$T_{211}$ by permuting the roles of the $x$-variables,
the $y$-variables and the $z$-variables. 
Similarly to the statement of Proposition \ref{prop_UH},
the form $T_{112}^{\otimes 2m}$ gives a direct sum of  $\Cc$-tensors 
with support isomorphic to $\braket{1,H,1}$, each component in the tensors
being isomorphic to $\braket{q^{2bm},q^{2(1-b)m},q^{2bm}}$.
The form $T_{121}^{\otimes 2m}$ gives a direct sum of  $\Cc$-tensors 
with support isomorphic to $\suppc(\braket{H,1,1})$, each component 
being isomorphic to 
$\braket{q^{2(1-b)m},q^{2bm},q^{2bm}}$.

Suppose that different constants are used to treat each of the three forms:
the forms $T_{112}$ and $T_{121}$ are processed with some constant $b$,
while $T_{211}$ is processed with another constant $\tilde b$. 
For any fixed values $a_{112},a_{211}$ and any $\epsilon>0$, the form $T_{112}^{\otimes a_{112}N}\otimes T_{112}^{\otimes a_{112}N}\otimes T_{211}^{\otimes a_{211}N}$
can then be used to construct a direct sum of
$$
\Omega\left(\frac{1}{N^3c'^{N\epsilon}}\cdot 
\left[\frac{2}{(2b)^{b}(1-b)^{1-b}}\right]^{2a_{112}N}\times \left[\frac{2}{(2\tilde b)^{\tilde b}(1-\tilde b)^{1-\tilde b}}\right]^{a_{211}N}
\right)
$$
$\Cc$-tensors, for some value $c'\ge 1$ depending only on $b$ and $\tilde b$.
Each of these $\Cc$-tensors has a support isomorphic to 
$\suppc(\braket{H_{112},H_{112},H_{211}})$, where 
\begin{eqnarray*}
H_{112}&=&\Omega\left(\frac{1}{\sqrt{N}}\cdot\left[(2b)^b(1-b)^{(1-b)}\right]^{a_{112}N}\right)\\
H_{211}&=&\Omega\left(\frac{1}{\sqrt{N}}\cdot\left[(2\tilde b)^{\tilde b}(1-\tilde b)^{(1-\tilde b)}\right]^{a_{211}N}\right).
\end{eqnarray*}
In all these $\Cc$-tensors, each component is isomorphic to the rectangular matrix multiplication 
$$\braket{q^{(a_{112}+ a_{211}\tilde b)N},q^{(a_{112}+ a_{211}\tilde b)N},q^{(2a_{112}b+ a_{211}(1-\tilde b))N}}.$$ 
We can then use Propositions \ref{prop_deg} and \ref{prop_Strassen} to convert each $\Cc$-tensor into a direct sum of at least
$\frac{3}{4}H_{112}\times \min(H_{112},H_{211})$ trilinear forms, each isomorphic to 
$\braket{q^{(a_{112}+ a_{211}\tilde b)N}\!,q^{(a_{112}+ a_{211}\tilde b)N}\!,q^{(2a_{112}b+ a_{211}(1-\tilde b))N}}$. 
We thus obtain the following result.

\begin{proposition}\label{cor_s}
Let $a_{112}$ and $a_{211}$ be any two positive constants.
Let $b$ and $\tilde b$ be any two constants such that $0.916027<b,\tilde b\le1$. 
Then there exists a constant $c'\ge 1$ such that, 
for any $\epsilon>0$, 
the form $$T_{112}^{\otimes a_{112}N}\otimes T_{121}^{\otimes a_{112}N}\otimes T_{211}^{\otimes a_{211}N}$$ can be 
converted into a direct sum of
$$
\Omega\left(\frac{1}{N^4c'^{N\epsilon}}\cdot \left[\frac{2^{2a_{112}+a_{211}}}
{\max\left(\left[(2b)^b(1-b)^{1-b}\right]^{a_{112}},\left[(2\tilde b)^{\tilde b}(1-\tilde b)^{1-\tilde b}\right]^{a_{211}}\right)}\right]^N\right)
$$
forms, each form being isomorphic to 
$\braket{q^{(a_{112}+ a_{211}\tilde b)N},q^{(a_{112}+ a_{211}\tilde b)N},q^{(2a_{112}b+a_{211}(1-\tilde b) )N}}.$
\end{proposition}
%
\subsection{Main theorem}\label{sub_together}

Let us define the following three quantities.
\begin{eqnarray*}
Q&=&(2q)^{a_{103}+a_{301}}\times (q^2+2)^{a_{202}}\times q^{a_{112}+a_{211}\tilde{b}}\\
R&=&(2q)^{2a_{013}}\times (q^2+2)^{a_{022}}\times q^{2a_{112}b+(1-\tilde{b})a_{211}}\\
\mathcal{M}&=&
\frac{2^{2a_{112}+a_{211}}}{A_0^{A_0}A_1^{A_1}A_2^{A_2}A_3^{A_3}A_4^{A_4}}\times
\frac{1}
{\max\left(\left[(2b)^b(1-b)^{1-b}\right]^{a_{112}},\left[(2\tilde b)^{\tilde b}(1-\tilde b)^{1-\tilde b}\right]^{a_{211}}\right)}.\\
\end{eqnarray*}
Our main theorem gives an upper bound on $\omega(1,1,k)$ that depends on these quantities.
\begin{theorem}\label{maintheorem}
Let $q$ be any positive integer and $b, \tilde b$ be such that $0.916027<b,\tilde b\le1$.
Let $a_{004}$, $a_{400}$, $a_{013}$, $a_{103}$, $a_{301}$, $a_{022}$, $a_{202}$,  $a_{112}$ and $a_{211}$ be any nine positive rational numbers 
satisfying the following three conditions:
\begin{itemize}
\item
$2a_{004}+a_{400}+2a_{013}+2a_{103}+2a_{301}+a_{022}+2a_{202}+2a_{112}+a_{211}=1$;
\item
$a_{013}a_{202}a_{112}=a_{103}a_{022}a_{211}$;
\item
$A_0^{A_0}A_1^{A_1}A_2^{A_2}A_3^{A_3}A_4^{A_4}\ge B_0^{B_0}B_1^{B_1}B_2^{B_2}B_3^{B_3}B_4^{B_4}$.
\end{itemize}
Then 
$$\mathcal{M}Q^{w(1,1,\frac{\log R}{\log Q})}\le (q+2)^2.$$
\end{theorem}
\begin{proof}
Let $\epsilon>0$ be an arbitrary positive value.
Let $N$ be a large integer and consider the trilinear form $(F_q\otimes F_q)^{\otimes N}$.
Theorem \ref{theorem_consresult} shows that this form can be used to obtain a direct sum of
$$ 
r_1=\Omega\left(\frac{1}{N^{10+8\epsilon}15^{N\epsilon}}\left[\frac{1}{A_0^{A_0}A_1^{A_1}A_2^{A_2}A_3^{A_3}A_4^{A_4}}\right]^{N}\right)
$$
forms, each isomorphic to
$$\bigotimes_{i,j,k:\: i+j+k=4}T_{ijk}^{\otimes a_{ijk}N}.$$

All the terms $T_{ijk}$ in this form, except $T_{112}$, $T_{121}$ and $T_{211}$, correspond to matrix multiplications and have been analyzed 
in Section \ref{section_CW}.
By Proposition \ref{cor_s} the part $T_{112}^{\otimes a_{121}N}\otimes T_{112}^{\otimes a_{112}N}\otimes T_{211}^{\otimes a_{211}N}$ 
can be used to obtain a direct sum of 
$$
r_2=\Omega\left(\frac{1}{N^4c'^{N\epsilon}}\cdot\left[
\frac{2^{2a_{112}+a_{211}}}
{\max\left(\left[(2b)^b(1-b)^{1-b}\right]^{a_{112}},\left[(2\tilde b)^{\tilde b}(1-\tilde b)^{1-\tilde b}\right]^{a_{211}}\right)}
\right]^N\right)
$$
matrix multiplications $\braket{q^{(a_{112}+ a_{211}\tilde b)N},q^{(a_{112}+ a_{211}\tilde b)N},q^{(2a_{112}b+(1-\tilde b) a_{211})N}}$.

This means that the trilinear form
$(F_q\otimes F_q)^{\otimes N}$ can be converted into a direct sum of 
$r_1r_2$
matrix multiplications $\braket{Q^N,Q^N,R^N}$. 
In other words:
$$
r_1r_2 \cdot \braket{Q^N,Q^N,R^N} \degen (F_q\otimes F_q)^{\otimes N}.
$$

Since $\underline{R}\left(F_q\otimes F_q\right)\le (q+2)^{2}$,
as mentioned in Section \ref{section_CW}, we know that 
$\underline{R}\left((F_q\otimes F_q)^{\otimes N}\right)\le (q+2)^{2N}$.
By Sch\"onhage's asymptotic sum inequality (Theorem~\ref{theorem_schonhage}) we then conclude that
$$
r_1r_2 \times Q^{N\omega(1,1,\frac{\log R}{\log Q})}\le (q+2)^{2N}.
$$
Taking the $N$-th root, we obtain:
\begin{eqnarray*}
\frac{1}{(15c')^\epsilon  N^{(14+8\epsilon)/N}}\times \mathcal{M} Q^{\omega(1,1,\frac{\log R}{\log Q})}
&\le& (q+2)^{2}.
\end{eqnarray*}
For any $\epsilon>0$ the above inequality holds for large enough integers $N$. 
By letting $N$ grow to infinity, and then letting $\epsilon$ decrease to zero, 
we conclude that $\mathcal{M} Q^{\omega(1,1,\frac{\log R}{\log Q})}\le (q+2)^{2}.$
\end{proof}

\section{Optimization}\label{sec_optimization}
In this section we use Theorem \ref{maintheorem} to derive numerical
upper bounds on the exponent of rectangular matrix multiplication,
and prove Theorem \ref{th_alpha}.

\subsection{Square matrix multiplication}
In this subsection we briefly show that our results give, for the exponent of square matrix multiplication,
the same upper bound as the bound obtained by Coppersmith and Winograd \cite{Coppersmith+90} .

Due to the symmetry of square matrix multiplication, 
we take $b=\tilde b$, $a_{400}=a_{004}$, $a_{103}=a_{013}=a_{301}$,
$a_{022}=a_{202}$ and $a_{112}=a_{211}$.
Then only six parameters remain, and the conditions
$a_{013}a_{202}a_{112}=a_{103}a_{022}a_{211}$ and
$A_0^{A_0}A_1^{A_1}A_2^{A_2}A_3^{A_3}A_4^{A_4}= B_0^{B_0}B_1^{B_1}B_2^{B_2}B_3^{B_3}B_4^{B_4}$
are immediately satisfied.

Theorem \ref{maintheorem} shows that
$$
\frac{8^{a_{112}}\times \left[(2q)^{2a_{013}}\times (q^2+2)^{a_{202}}\times q^{(1+b)a_{112}}\right]^{\omega(1,1,1)}}
{A_0^{A_0}A_1^{A_1}A_2^{A_2}A_3^{A_3}A_4^{A_4}\times\left[(2b)^b(1-b)^{(1-b)}\right]^{a_{112}}}
\le (q+2)^2.
$$
By choosing $q=6$, $b=0.9724317$, $a_{103}=0.012506$, $a_{202}=0.102546$, $a_{112}=0.205542$ and 
$a_{004}=0.0007/3$, we obtain the upper bound $\omega(1,1,1)<2.375477$. 

This upper bound on the exponent of square matrix multiplication 
is exactly the same value as in~\cite{Coppersmith+90}.
This is not a coincidence.
Indeed, by setting $b=\frac{q^{\omega(1,1,1)}}{q^{\omega(1,1,1)}+2}$, which is larger
than 0.916027 for $q\ge 5$, we obtain
$$
\frac{q^{b\omega(1,1,1)}}{(2b)^b(1-b)^{(1-b)}}=\frac{q^{\omega(1,1,1)}+2}{2}
$$
and our inequality becomes
$$
\frac{(2q)^{2a_{013}\omega(1,1,1)}\times (q^2+2)^{a_{202}\omega(1,1,1)}\times (4q^{\omega(1,1,1)}(q^{\omega(1,1,1)}+2))^{a_{112}}}
{A_0^{A_0}A_1^{A_1}A_2^{A_2}A_3^{A_3}A_4^{A_4}}
\le (q+2)^2,
$$
which is exactly the same optimization problem as in Section 8 of \cite{Coppersmith+90}.

\subsection{Rectangular matrix multiplication}
In this subsection we explain how to use Theorem \ref{maintheorem} to derive an upper bound on $\omega(1,1,k)$ for 
an arbitrary value $k$, and show how to obtain the results stated in Table \ref{table_results}
and Figure \ref{fig_results}. 

We use the following strategy.
We take a positive integer $q$, seven positive rational 
numbers
$a_{400}$, $a_{103}$, $a_{301}$, $a_{022}$, $a_{202}$,
$a_{112}$ and $a_{211}$, and two values $b,\tilde b$ such that
$0.916027<b,\tilde b\le 1$.
We then fix 
$$
a_{013}=\frac{a_{103}a_{022}a_{211}}{a_{202}a_{112}}
$$
and 
$$
a_{004}=\frac{1-(a_{400}+2a_{013}+2a_{103}+2a_{301}+a_{022}+2a_{202}+2a_{112}+a_{211})}{2}.
$$
The conditions that have to be satisfied are:
\begin{itemize}
\item
$0<a_{004},a_{013}\le 1$;
\item
$A_0^{A_0}A_1^{A_1}A_2^{A_2}A_3^{A_3}A_4^{A_4}\ge B_0^{B_0}B_1^{B_1}B_2^{B_2}B_3^{B_3}B_4^{B_4}$.
\end{itemize}
If these conditions are satisfied, by Theorem \ref{maintheorem} this gives the upper bound
$$
\omega\left(1,1,\frac{\log R}{\log Q}\right)\le \frac{2\log (q+2)-\log \mathcal{M}}{\log Q}.
$$

The above discussion reduces the problem of finding an upper bound on $\omega(1,1,k)$ to 
solving a nonlinear optimization problem. The 
upper bounds presented in Table~\ref{table_results} are obtained precisely by solving this optimization 
problem using Maple. 
We show exact values of the parameters proving that 
$\omega(1,1,0.5302)<2.060396$,
$\omega(1,1,0.75)< 2.190087$ and $\omega(1,1,2)< 3.256689$ in Table~\ref{table_values}.

\begin{table}
\begin{center}
\begin{tabular}{ | c | c| c|c|}
  \hline 
$q$ &5& 5&   6 \\
$b$            & 0.984599222& 0.968978515  & 0.94866036 \\
$\tilde b$   &0.919886704&  0.938616630         &   0.99996514 \\   
$a_{400}$ &0.004942000&  0.001498500                &   0.00000090\\
$a_{103}$ &0.010965995&  0.014456894       &   0.01553556 \\
$a_{301}$ &0.055710210&  0.031215255       &   0.00079349 \\
$a_{022}$  &0.037622078&  0.065869083     &   0.22704392 \\
$a_{202}$  &0.138698196&  0.118190058       &   0.05836108 \\
$a_{112}$ &0.145715589&  0.178553843   &   0.20388121 \\
$a_{211}$ &0.245013049&  0.226835534   &   0.13394891 \\
\hline
$a_{004}$  &0.00011...&  0.000246... &  0.001224... \\
$a_{013}$  &0.00500...&  0.010235...   & 0.039707... \\
$A_0^{A_0}A_1^{A_1}A_2^{A_2}A_3^{A_3}A_4^{A_4}$  &0.326588...& 0.3265988... & 0.339123647... \\
$B_0^{B_0}B_1^{B_1}B_2^{B_2}B_3^{B_3}B_4^{B_4}$  &0.326587...& 0.3265987... & 0.339123642... \\ 
\hline
$\log R/\log Q$  &0.530200005...& 0.750000001...  & 2.00000004... \\
$(2\log(q+2)-\log\mathcal{M})/\log Q$  &2.060395...&2.190086...  & 3.256688...  \\            
  \hline  
\end{tabular}
\caption{Three solutions for our optimization problem.
The first ten rows give (exact) values of the ten parameters. 
The numeral values of the next four rows show that the three conditions 
$0<a_{004},a_{013}\le 1$ and $A_0^{A_0}A_1^{A_1}A_2^{A_2}A_3^{A_3}A_4^{A_4}\ge B_0^{B_0}B_1^{B_1}B_2^{B_2}B_3^{B_3}B_4^{B_4}$ are satisfied.
The numerical values of the last two rows show that $\omega(1,1,0.5302)<2.060396$,
$\omega(1,1,0.75)< 2.190087$ and $\omega(1,1,2)< 3.256689$.
\label{table_values}}
\end{center}
\end{table}

\subsection{The value {\boldmath$\alpha$}}

In this subsection we describe how to use Theorem \ref{maintheorem} 
to obtain a lower bound on the value $\alpha$,
the largest value such that the product of an $n\times n^\alpha$ matrix by an 
$n^\alpha\times n$ matrix can be computed  with $O(n^{2+\epsilon})$ arithmetic operations for any $\epsilon>0$.
The analysis is more delicate than in the previous subsection, since we will need to exhibit parameters
such that $\mathcal{M} Q^2=(q+2)^2$, with an equality rather than an inequality,
and is done by finding analytically the optimal values of all but a few parameters.

Let $q$ be an integer such that $q\ge 5$.
For convenience, we will write $\kappa=1/(q+2)^2$. 
Let $a_{112}$ and $a_{211}$ be any rational numbers such that 
$0<a_{112}< q\kappa$ and $0< a_{211}< (q^2+2)\kappa$.
We set the parameters $b$, $\tilde b$, $a_{004}$, $a_{103}$, $a_{202}$ and $a_{301}$ as follows:
\begin{eqnarray*}
b&=&1\\
\tilde b&=& q^2/(q^2+2)\\
a_{400}&=&\kappa\\
a_{103}&=&q\kappa-a_{112}\\
a_{202}&=&\left((q^2+2)\kappa-a_{211}\right)/2\\
a_{301}&=&q\kappa.
\end{eqnarray*}

Putting these values in the formula for $Q$, we obtain:
\begin{eqnarray*}
Q&=&\left(2q\right)^{q\kappa+a_{103}}\times \left(q^2+2\right)^{\frac{(q^2+2)\kappa-a_{211}}{2}}\times q^{a_{112}+q^2a_{211}/(q^2+2)}\\
&=&(2q)^{2q\kappa}\times (q^2+2)^\frac{(q^2+2)\kappa}{2}\times 2^{-a_{112}}\times \left(\frac{q^{q^2/(q^2+2)}}{\sqrt{q^2+2}}\right)^{a_{211}}.
\end{eqnarray*}
Observe that $A_1=A_3=2q\kappa$, $A_2=(q^2+2)\kappa$, $A_4=\kappa$ and $A_0=1-(A_1+A_2+A_3+A_4)=\kappa$. 
Then we obtain the following equality.
\begin{eqnarray*}
\frac{1}{A_0^{A_0}A_1^{A_1}A_2^{A_2}A_3^{A_3}A_4^{A_4}} 
&=&\frac{(q+2)^2}{(2q)^{4q\kappa}(q^2+2)^{(q^2+2)\kappa}}
\end{eqnarray*}
The following lemma shows that, when $a_{112}$ is small enough, the condition $\mathcal{M} Q^2=(q+2)^2$ is satisfied.
\begin{lemma}
Suppose that 
\begin{equation}\label{cond1}
a_{112}\le \left(1+\frac{2q^2}{q^2+2}\log_2(q)-\log_2(q^2+2)\right)a_{211}.
\end{equation}
Then 
$\mathcal{M} Q^2=(q+2)^2$.
\end{lemma}
\begin{proof}
Our choice for $b$ and $\tilde b$ gives
\begin{eqnarray*}
\left[(2b)^b(1-b)^{1-b}\right]^{a_{112}}&=&2^{a_{112}}\\
\left[(2\tilde b)^{\tilde b}(1-\tilde b)^{1-\tilde b}\right]^{a_{211}}&=&\left[\frac{2}{q^2+2}\cdot q^{\frac{2q^2}{q^2+2}}\right]^{a_{211}}.
\end{eqnarray*}
Inequality (\ref{cond1}) then implies that $\left[(2b)^b(1-b)^{1-b}\right]^{a_{112}}\le 
\left[(2\tilde b)^{\tilde b}(1-\tilde b)^{1-\tilde b}\right]^{a_{211}}$.
In consequence,
$$
\mathcal{M}=
\frac{(q+2)^2}{(2q)^{4q\kappa}(q^2+2)^{(q^2+2)\kappa}}\times 4^{a_{112}}\times \left[\frac{q^2+2}{q^{2q^2/(q^2+2)}}\right]^{a_{211}},
$$
which gives $\mathcal{M} Q^2=(q+2)^2$.
\end{proof}

We now explain how to determine the three remaining parameters $a_{004}$, $a_{013}$ and $a_{022}$.
Remember that the parameters should satisfy the equalities 
$$a_{013}=\frac{a_{103}a_{211}}{a_{202}a_{112}}a_{022}$$
and
$$2a_{004}+a_{400}+2a_{013}+2a_{103}+2a_{301}+a_{022}+2a_{202}+2a_{112}+a_{211}=1.$$
From our choice of parameters, the second equality can be rewritten as $2a_{004}+2a_{013}+a_{022}=\kappa$.
Since the parameter $a_{004}$ should be positive, we obtain the condition
\begin{equation}\label{cond2}
\left(\frac{4(q\kappa-a_{112})a_{211}}{((q^2+2)\kappa-a_{211})a_{112}}+1\right)a_{022}< \kappa.
\end{equation}
If $a_{022}$, $a_{112}$ and $a_{211}$ satisfy this inequality, then the parameter $a_{004}$ is fixed:
$$
a_{004}=\left(\kappa-\left(\frac{4(q\kappa-a_{112})a_{211}}{((q^2+2)\kappa-a_{211})a_{112}}+1\right)a_{022}\right)/2.
$$
Note that Inequality (\ref{cond2}) forces the value $a_{013}$ to be at most $1$.

All the values are thus determined by the choice of $q$, $a_{022}$, $a_{112}$ and $a_{211}$. In particular,
we obtain
\begin{eqnarray*}
R
&=&\left(2q\right)^{\frac{4(q\kappa-a_{112})a_{211}}{((q^2+2)\kappa-a_{211})a_{112}}}\times \left(q^2+2\right)^{a_{022}}\times q^{2a_{112}+\frac{2}{q^2+2}a_{211}}.
\end{eqnarray*}
We can similarly express the values of $B_0$, $B_1$, $B_2$, $B_3$ and $B_4$ in function of these four parameters.
We then want to solve the following optimization problem.

\begin{center}
    \begin{tabular}[h!]{|p{0.8\textwidth}|}
    \hline\vspace{1mm}
        Maximize $\frac{\log R}{\log Q}$ subject to
\begin{itemize}
\item $0\le a_{022}\le 1$, $0<a_{112}\le 5\kappa$ and $0\le a_{211}\le (q^2+2)\kappa$;
\item $q$ is an integer such that $q\ge 5$;
\item Inequalities (\ref{cond1}) and (\ref{cond2}) hold;
\item $\frac{(2q)^{4q\kappa}(q^2+2)^{(q^2+2)\kappa}}{(q+2)^2}\ge B_0^{B_0}B_1^{B_1}B_2^{B_2}B_3^{B_3}B_4^{B_4}$.
\end{itemize}
        \\\hline
    \end{tabular}
\end{center}

By taking the values $q=5$, $a_{022}=0.0174853$, $a_{112}=0.0945442$ and $a_{211}=0.1773724$, 
we obtain 
the value $\alpha\ge \frac{\log R}{\log Q}>0.30298$. These parameters satisfy all the constraints. 
We obtain in particular the following numerical values.
\begin{eqnarray*}
\frac{(2q)^{4q\kappa}(q^2+2)^{(q^2+2)\kappa}}{(q+2)^2}&=&0.3211277...\\
B_0^{B_0}B_1^{B_1}B_2^{B_2}B_3^{B_3}B_4^{B_4}&=&0.3211276....\\
R&=&1.475744...\\
Q&=&3.612672...
\end{eqnarray*}

A more precise lower bound on $\alpha$ can be found using optimization software and high precision 
arithmetic.
Using Maple and truncating the result of the optimization after the 25th digit, we find that
for $q=5$ the values
\begin{eqnarray*}
a022 &=& 0.0174853267797595451457284\\
a112 &=& 0.0945442542111395375830367\\
a211 &=& 0.1773724081899825630904504
\end{eqnarray*}
give the lower bound
$$
\alpha>     0.3029805825293869820274449.
$$

\section*{Acknowledgments}
The author is grateful to Virginia Vassilevska Williams and Ryan Williams
for helpful correspondence about Andrew Stothers' work, and to 
Virginia Vassilevska Williams for suggesting that the next step is to use higher 
tensor powers of the basic construction to improve rectangular 
matrix multiplication. He also
acknowledges support from the JSPS and the MEXT,
under the grant-in-aids 
Nos.~22800006 and 24700005.
\bibliographystyle{acm}
\bibliography{LeGallMM}
\end{document}